\documentclass[journal,letterpaper]{IEEEtran}

\usepackage{cite}
\usepackage{amsmath,amssymb,amsfonts}
\usepackage{algpseudocode}
\usepackage{algorithm}
\usepackage{graphicx}
\usepackage{textcomp}
\usepackage{comment}
\usepackage{braket}
\usepackage{bbold}
\usepackage{url}
\usepackage{xcolor} 

\usepackage[T1]{fontenc}
\usepackage{amsthm}
\theoremstyle{plain}
\newtheorem{theorem}{Theorem}
\newtheorem{corollary}{Corollary}
\newtheorem{lemma}{Lemma}
\theoremstyle{definition}
\newtheorem{definition}{Definition}

\algrenewcommand\textproc{}%
\DeclareMathOperator*{\argmax}{arg\,max}

\def\BibTeX{{\rm B\kern-.05em{\sc i\kern-.025em b}\kern-.08em
    T\kern-.1667em\lower.7ex\hbox{E}\kern-.125emX}}

\begin{document}

\title{Routing in Non-Isotonic Quantum Networks}


\author{%
\IEEEauthorblockN{%
Maxwell Tang\IEEEauthorrefmark{2}\IEEEauthorrefmark{1},
Garrett Hinkley\IEEEauthorrefmark{2}\IEEEauthorrefmark{1},
Kenneth Goodenough\IEEEauthorrefmark{2}\IEEEauthorrefmark{3},
Stefan Krastanov\IEEEauthorrefmark{2},
Guus Avis\IEEEauthorrefmark{2}}

\IEEEauthorblockA{\IEEEauthorrefmark{2}Manning College of Information and Computer Sciences,\\
University of Massachusetts Amherst, Amherst, Massachusetts 01003, USA}

\IEEEauthorblockA{\IEEEauthorrefmark{3}Naturwissenschaftlich-Technische Fakult{\"a}t, Universit{\"a}t Siegen,\\
Walter-Flex-Stra{\ss}e 3, 57068 Siegen, Germany}
\thanks{\IEEEauthorrefmark{1}Maxwell Tang and Garrett Hinkley contributed equally to this work.}
\thanks{We thank Eva Peet for help with creating Fig.~\ref{fig:skr_not_isotonic}.}
\thanks{S.K.\ and G.A.\ acknowledge support from NSF CQN Grant No.~1941583, NSF Grant No.~2346089, and NSF Grant No.~2402861.}
\thanks{K.G.\ acknowledges support by the Alexander von Humboldt foundation.}
\thanks{Corresponding author: Guus Avis (email: guusavis@hotmail.com).}
}

\maketitle

\begin{abstract}
Optimal routing in quantum-repeater networks requires finding the best path that connects a pair of end nodes.
Most previous work on routing in quantum networks assumes utility functions that are isotonic,
meaning that the ordering of two paths does not change when extending both with the same edge.
However, we show that utility functions that take into account both the rate and quality of the entanglement generation (e.g., the secret-key rate) are often non-isotonic.
This makes pathfinding difficult as classical algorithms such as Dijkstra’s become unsuitable, with the state of the art for quantum networks being an exhaustive search over all possible paths.
In this work we present improved algorithms.
First, we present two best-first-search algorithms that use destination-aware merit functions for faster convergence.
One of these provably finds the best path, while the other uses heuristics to achieve an effectively sublinear scaling of the query count in the network size while in practice always finding a close-to-optimal path.
Second, we present metaheuristic algorithms (simulated annealing and a genetic algorithm) that enable tuning a tradeoff between path quality and computational overhead.
While we focus on swap-ASAP quantum repeaters for concreteness, our algorithms are readily generalized to different repeater schemes and models.
\end{abstract}

\section{Introduction}
\label{sec:introduction}
A quantum network enables the creation of entangled quantum states amongst its end nodes.
Such states can be used for a plethora of applications, such as the transmission of quantum information through teleportation~\cite{bennett1993teleporting, bouwmeester1997experimental, pirandola2015advances}, secure communication through quantum key distribution (QKD)~\cite{bennett2014quantum, bennett1992a}, distributed sensing~\cite{bennett2014quantum}, blind quantum computing~\cite{broadbent2009universal, fitzsimons2017private, van2024hardware}, and distributed quantum computation~\cite{buhrman2003distributed, van2016path}.
In terrestrial quantum networks, entanglement is typically distributed by exchanging photons through optical fiber, but this incurs losses that are exponential in the fiber length.
To mitigate these losses, a type of auxiliary network node is required, called a quantum repeater~\cite{azuma2023quantum}.
While practical quantum repeaters have not yet been realized, there have been several promising proof-of-principle demonstrations~\cite{li2019experimental, langenfeld2021quantum, pu2021experimental}.

To create entanglement between two specific end nodes in a larger quantum-repeater network, a path connecting those nodes is required, as illustrated in Fig.~\ref{fig:best_path_example}. 
\begin{figure}[t!]
\centering
\includegraphics[width=.99\columnwidth]{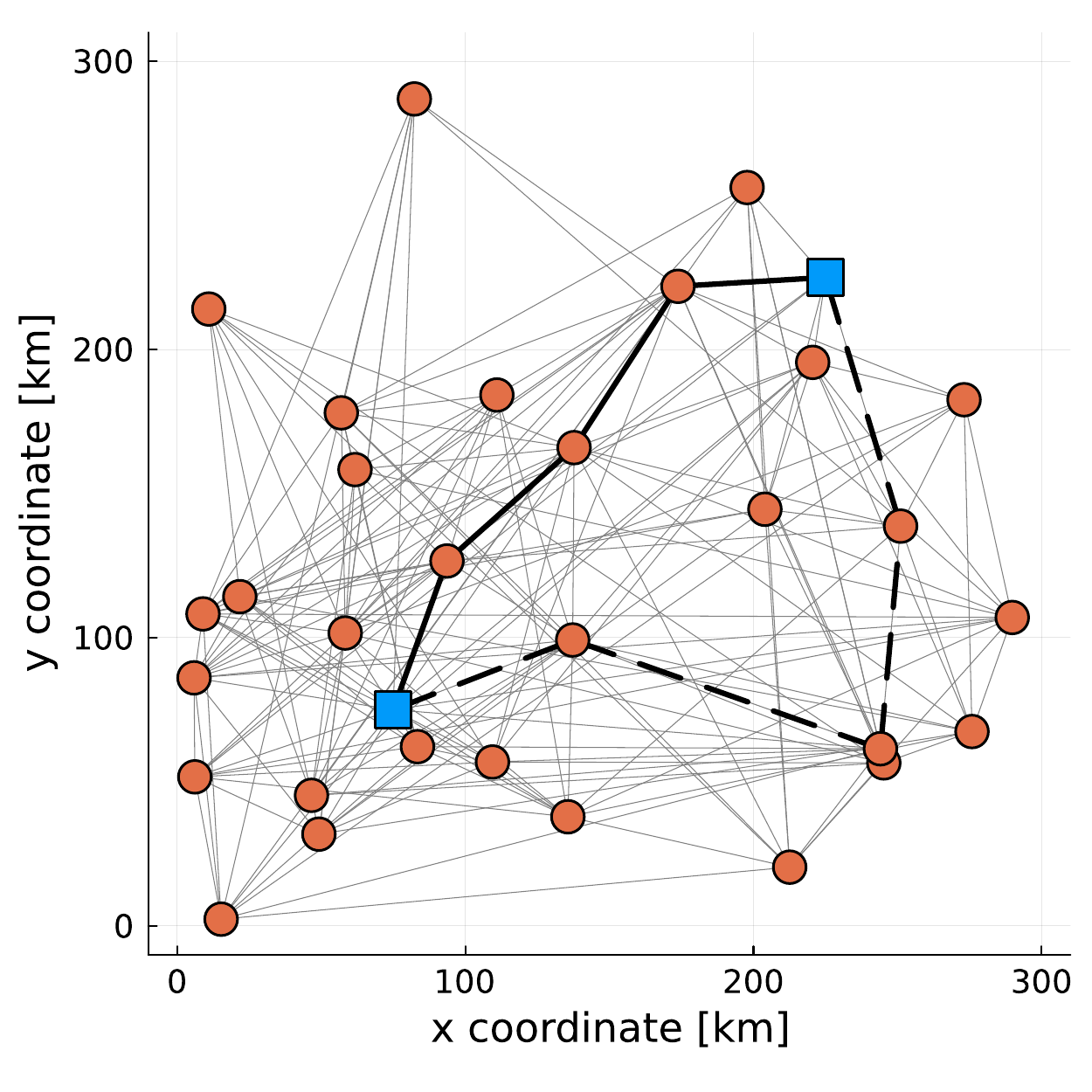}
\caption{
    An example of a randomly generated quantum network (procedure explained in Sec.~\ref{sec:problem_statement}), where blue squares indicate the end nodes that want to share entanglement and the orange circles represent quantum repeaters.
    The solid thick black line is the best path through the network, realizing a secret-key rate for quantum-key distribution (see Sec.~\ref{sec:problem_statement}) of 6.4 Hz, found using our BeFS-EXACT algorithm (see Sec.~\ref{sec:deterministic_algs}).
    On the other hand, the dashed thick black line is the path found using an algorithm (the extended Dijkstra's algorithm~\cite{shi2020a}) that assumes an isotonic utility function, which the secret-key rate is not in our repeater model.
    Over this path, a secret key can only be created at 0.17 Hz.
    The model parameters (explained in Sec.~\ref{sec:problem_statement}) are a link fidelity of $F = 0.96$ and coherence time of $T = 10$ s.
    \label{fig:best_path_example}
}
\end{figure}
Assigning paths is the responsibility of a routing protocol~\cite{abane2024}, and has recently received considerable interest~\cite{caleffi2017, schoute2016shortcuts, chakraborty2020entanglement, coutinho2023entanglement, victora2020purification, chakraborty2019distributed,van2013path, shi2020concurrent}.

In order for routing protocols to achieve the best possible network performance, they first need to know what the best possible paths through the network are.
Finding good paths between a pair of network nodes is a task that we refer to as pathfinding here, and finding pathfinding algorithms that are both efficient and effective when applied to quantum-repeater networks is the main goal of this paper.
We note that while pathfinding is essential to routing, the responsibility of routing protocols is broader than just identifying good paths.
For instance, they also need to adapt dynamically to changes in the network, organize how resources are shared, and deal with communication delays between nodes.

The contributions of this paper are twofold.
As a first contribution, we demonstrate that, for common utility functions and quantum-repeater models, the standard assumption of \textit{isotonicity} is strongly broken in quantum networks. 
A utility function is isotonic if the relative ordering between two paths cannot be reversed by appending the same edge to both those paths (see Def.~\ref{def:isotonicity} for an exact definition).
When a utility function is not isotonic, standard pathfinding algorithms like Dijkstra's algorithm are not guaranteed to find the optimal path~\cite{sobrinho2005} (see Fig.~\ref{fig:best_path_example} for an example).
As a second contribution, we present pathfinding algorithms for quantum networks when isotonicity does not hold.

The emergence of non-isotonicity in quantum networks was first discovered by Caleffi~\cite{caleffi2017}.
In that work, an analytical model was developed for the rate at which entanglement can be distributed along a network path.
Caleffi observed that this model violates isotonicity and that hence traditional algorithms are unable to find the optimal path.
He addressed this problem by proposing a pathfinding algorithm based on the explicit enumeration and evaluation of all possible paths.
We remark that the non-isotonic behavior of Caleffi's model is the result of a particular simplifying assumption used in that work,
namely that the probabilistic process of entanglement generation along an edge takes a constant, deterministic amount of time.
As a result, the rate along any path is given by the rate of its slowest edge, see, e.g.,~\cite{pirandola2019end} and~\cite{coutinho2023entanglement} for similar models.

As such, isotonicity is violated in Caleffi's model in the weakest way possible: two paths with different utility values can become equally good by adding an edge (namely, in case that edge becomes the bottleneck in both paths), but their ordering can never be reversed.
In fact, the model is still isotonic according to the definition we adhere to in this paper, and we point out that pathfinding in Caleffi's model could be solved efficiently since it reduces to the widest path problem~\cite{widestpathproblem}.

We also point out the work from~\cite{coutinho2023entanglement}, which showed isotonicity of the fidelity of the end-to-end state for two models: the single-e-bit and flow-distribution models.
The former requires all elementary links to succeed on the first attempt (and thus does not scale); the latter is similar to the model used by Caleffi and Pirandola et al.~\cite{pirandola2019end}.

We show here that non-isotonicity naturally arises when considering realistic distribution- and error models with a practical figure of merit.
In particular, we assume that 1) entanglement is distributed using a so-called \emph{swap-ASAP protocol}~\cite{kamin2023, goodenough2025} without distillation, and 2) we consider the secret-key rate of quantum key distribution, which captures the rate-fidelity trade-off in such a way that it quantifies the usefulness of a path for one of the most popular quantum-network applications.

The failure of isotonicity occurs for two reasons (see Sec.~\ref{sec:problem_statement} for more detail).
First, while individually the rate and fidelity can be isotonic for a given entanglement-distribution model, metrics like the secret-key rate need not be.
This arises from the trade-off between rate and quality; generating a worse-quality state at a higher rate might---depending on the setting---be better or worse than preparing a higher-quality state at a lower rate. Secondly, it is not only the trade-off between rate and fidelity which leads to non-isotonicity. In fact, even the rate and fidelity individually are themselves non-isotonic for swap-ASAP repeater chains (see Sec.~\ref{sec:isotonicity}).

As such, routing algorithms for quantum networks must account for non-isotonicity.
To enable routing in quantum networks without resorting to explicit enumeration~\cite{caleffi2017} ---which has a super-exponential complexity in the network size---we present two types of pathfinding algorithms that work even when utility functions are non-isotonic.
As our first result we present best-first-search algorithms inspired by the A* algorithm~\cite{foead2021systematic}.
One of these algorithms (named BeFS-EXACT) is guaranteed to always find the best path in the case of swap-ASAP repeaters, while the second (named BeFS-HEURISTIC) uses heuristics that we expect to be effective for most repeater types.
Importantly, in our experiments (see Sec.~\ref{sec:algorithm_comparisons}), BeFS-HEURISTIC almost always finds the best path with a query count that is effectively sublinear in the network size.
As our second result, we present metaheuristic algorithms (genetic algorithms and simulated annealing) that are suitable for any non-isotonic utility function and allow explicitly tuning a tradeoff between path quality and computational overhead.

Our paper is structured as follows.
In Sec.~\ref{sec:problem_statement} we introduce the utility function (secret-key rate) and repeater model (swap-ASAP repeaters) that we will use throughout the rest of the paper, discuss the importance of isotonicity, and explicitly show that the utility function is not isotonic.
In Sec.~\ref{sec:deterministic_algs} we present our first algorithms, which are based on a best-first search.
Our second set of algorithms, based on metaheuristics, is presented in Sec.~\ref{sec:metaheuristic_algs}.
We compare the performance of the different algorithms in Sec.~\ref{sec:algorithm_comparisons} and conclude in Sec.~\ref{sec:conclusion}.

\section{Problem Statement}
\label{sec:problem_statement}

Let $G$ be an undirected weighted graph, and let $f$ be a utility function that maps each path in $G$ to a real number.
The goal of pathfinding then is the following: given nodes $s$ and $t$, find a path between $s$ and $t$ that maximizes $f$.
For our purposes $G$ represents a physical quantum network with quantum repeaters at its nodes, such that the weight of each edge corresponds to the length of the corresponding communication channel (e.g., an optical fiber).
Here, we make the simplifying assumption that all nodes and channels have exactly the same properties, except for the fact that different channels have different lengths.
The utility function then depends only on the weights of the edges in a path, and hence can be considered a function that maps sequences of real numbers to a single real number.

\subsection{Utility Function}
\label{sec:utility_function}

In practice, what the utility function looks like depends on two things.
First, it depends on the performance metric used to quantify how good a path is at facilitating quantum communication between its end nodes.
Second, it depends on the architecture (both their physical design and the protocols used to operate them) and quality of the quantum repeaters in the network.
While our results are applicable to different performance metrics and different types of quantum repeaters (see Sec.~\ref{sec:conclusion}), for concreteness we here focus on one specific utility function:
the secret-key rate (SKR) of the BBM92 quantum-key-distribution
protocol~\cite{bennett1992a} for swap-as-soon-as-possible (swap-ASAP) quantum repeaters.

The SKR is a popular choice for a performance metric to optimize and benchmark quantum repeaters~\cite{rozpkedek2018parameter, rozpkedek2019near} as it quantifies how good they are at supporting a concrete quantum-network application, namely the generation of a shared secret key between two end nodes; the SKR is defined as the amount of key that can be generated per time unit.
Moreover, it conveniently combines the physical rate at which entangled pairs can be established by the network and the level of noise into a single number.
It can be written as
\begin{equation} \label{eq:SKR}
\text{SKR} = R \times \text{SKF},
\end{equation}
where $R$ is the rate of entanglement distribution (number of entangled pairs per time unit) and SKF is the secret-key fraction, i.e., the number of secret bits that can be distilled per entangled pair.
The SKF is a number between 0 and 1 that depends on the quality of the entangled states.
In the limit of asymptotically large keys when running the BBM92 protocol, it is given by~\cite{shor2000}
\begin{equation} \label{eq:SKF}
\text{SKF} = \max\left(0, 1 - h(Q_x) - h(Q_Z)\right),
\end{equation}
where $h(x) =-x\log_2(x) - (1-x)\log_2(1-x)$ is the binary entropy function and $Q_X$ ($Q_Z$) is the quantum-bit error rate (QBER) for measurements in the Pauli $X$ ($Z$) basis.

Swap-ASAP quantum repeaters~\cite{coopmans2021, goodenough2025, avis2022a, avis2024, avis2025, haldar2025, inesta2023, kamin2023} are repeaters that first create entanglement with each neighboring node, and then, as soon as both entangled states are present, perform entanglement swapping.
Entanglement swapping is an entangling measurement such that if the repeater had an entangled state with node $A$ and an entangled state with node $B$, $A$ and $B$ now share an entangled state.
When every repeater has performed entanglement swapping, the end nodes are successfully entangled.
We assume that each repeater has one qubit register per neighboring node and that an attempt at entanglement generation over an edge of length $L$ has success probability $10^{- \frac \alpha {10}  L}$ with $\alpha = 0.2 \, \text{km}^{-1}$ (equal to the probability that a photon is successfully transmitted through a fiber of length $L$ with 0.2 dB/km attenuation losses).
Moreover, we assume that upon success a Werner state with fidelity $F$ is created, which has density matrix
\begin{equation}
\frac{4F - 1}{3} \ket{\phi^+}\bra{\phi^+} + \frac 1 3 (1 - F) \mathbb 1_4
\end{equation}
where $\mathbb 1_n$ is the $n$-dimensional identity matrix, and that each time unit a qubit with index $k$ is stored for $t$ time in a memory with coherence time $T$ it undergoes depolarizing noise following
\begin{equation} \label{eq:memory_noise}
\rho \to e^{-\frac t {T}} \rho + (1 - e^{-\frac t {T}}) \frac{\mathbb 1_2} 2 \text{Tr}_k(\rho),
\end{equation}
where $\text{Tr}_k$ is the partial trace over qubit $k$.
End nodes measure their qubits directly after creating entanglement with their neighboring nodes to perform QKD, and hence do not suffer from memory decoherence.

We model entanglement generation over an edge of length $L$ as having duration $L/c$, where $c=200,000$ km/s is the speed of light in fiber (approximately two thirds the speed of light in vacuum).
However, we assume here that entanglement generation on shorter edges is purposefully delayed such that all edges have the same effective attempt duration, namely, $L_\text{max}/c$, where $L_\text{max}$ is the length of the longest edge in the chain.
Synchronizing attempts will reduce the entangling rate but can sometimes reduce qubit storage times.
Whether the SKR is higher when synchronizing or when not synchronizing depends on the configuration of the repeater chain.
We assume synchronized attempts, as it simplifies our model and makes it easier to prove bounds (see Appendix~\ref{app:bounds}).

We remark that our model does not closely capture any specific real-world implementation of swap-ASAP repeaters (see, e.g., \cite{avis2022a} for a more detailed model), but has rather been chosen for being relatively simple while nonetheless capturing the most critical properties of any real implementation.
This includes the exponential decrease of the success probability with the link length common to all heralded-entanglement-generation schemes using optical fiber (including single-click~\cite{cabrillo1999}, double-click~\cite{barrett2005} and direct-transmission schemes), as well as the fact that entangled links will already be noisy when first generated and then accumulate more noise the longer they are stored in memory.
Moreover, note that there are two free parameters in the model we have introduced: initial-state fidelity $F$ and the coherence time $T$.
That means that we now have a family of utility functions: $\text{SKR}_{F, T}$, which is the function that calculates the SKR of a path (which is effectively defined by a sequence of lengths) assuming a specific value for the fidelity and a specific value for the coherence time.
We evaluate the function $\text{SKR}_{F, T}$ numerically using the Julia package QuantumNetworkRecipes.jl~\cite{quantum_network_recipes.jl}, which was first introduced in \cite{avis2025}.
This implementation simulates the creation of a certain number of entangled pairs and estimates $\text{SKR}_{F, T}$ by dividing the sum of the secret-key fractions by the total time.
In our experiments, we take 10000 samples for each measurement.

\subsection{Isotonicity}
\label{sec:isotonicity}

Now, we make an important observation:
the utility function $\text{SKR}_{F, T}$ is not (always) isotonic (defined below).
This is problematic for pathfinding, as typical pathfinding algorithms~\cite{dijkstra1959, hart1968, yen1970} heavily rely on isotonicity~\cite{sobrinho2005}.
In particular, there exists a version of Dijkstra's algorithm (given in \cite{shi2020a} under the name ``extended Dijkstra's algorithm'') that efficiently finds optimal paths for any utility function that is both monotonic and isotonic (but not necessarily additive, as is assumed in the classical version of Dijkstra's algorithm~\cite{dijkstra1959}).
To understand this, we give formal definitions of monotonicity and isotonicity below.
First we introduce another piece of notation: if $p = (e_1, e_2, ..., e_k)$ is a path ending at node $v$ and $e$ is an edge connected to node $v$, then a new path can be formed by appending $e$ to $p$, which we write as $p \oplus e = (e_1, e_2, ..., e_k, e)$. For notational simplicity we will sometimes denote paths by their constituent vertices as opposed to their constituent edges; the intended meaning will be clear from the context.
\begin{definition} \label{def:monotonicity}
\emph{Monotonicity.}
A utility function $f$ that maps paths to real numbers is monotonic if and only if for any path p that ends at some node $v$ and every edge $e$ connected to node $v$ it holds that $f(p \oplus e) \leq f(p)$.
\end{definition}
That is, when a utility function is monotonic, a path can never be improved by adding edges to it.
\begin{definition} \label{def:isotonicity}
\emph{Isotonicity.}
A utility function $f$ that maps paths to real numbers is isotonic if and only if for any two paths $p_1$ and $p_2$ that end at the same node $v$ and for every edge $e$ connected to node $v$, $f(p_1) \geq f(p_2)$ implies $f(p_1 \oplus e) \geq f(p_2 \oplus e)$.
\end{definition}
That is, appending the same edge to two paths cannot change the relative ordering of those two paths with respect to an isotonic utility function.
Note that Caleffi~\cite{caleffi2017} defines isotonicity using a strict inequality; we commented on this discrepancy in Sec.~\ref{sec:introduction}.

We now explain why isotonicity allows for efficient algorithms such as Dijkstra's. Such an algorithm repeatedly extends partial paths (which we refer to as path prefixes) by appending edges to them (this procedure is discussed in more detail in Sec.~\ref{sec:deterministic_algs}).
When the utility function is isotonic, one only ever needs to extend the best path prefix leading up to a specific node, as all the other prefixes will remain suboptimal after the extensions.
For non-isotonic utility functions this is not the case, and one needs to track and extend many more path prefixes in order to find the best final path (one does not necessarily need to extend all path prefixes if the utility function is monotonic, i.e., decreasing under edge appendage; see Sec.~\ref{sec:deterministic_algs}).

We prove that the utility function $\text{SKR}_{F,T}$ is not isotonic for at least some values of $F$ and $T$ by providing a concrete counterexample, which can be found in Fig.~\ref{fig:skr_not_isotonic}.
\begin{figure}[t!]
\centering
\includegraphics[width=.99\columnwidth]{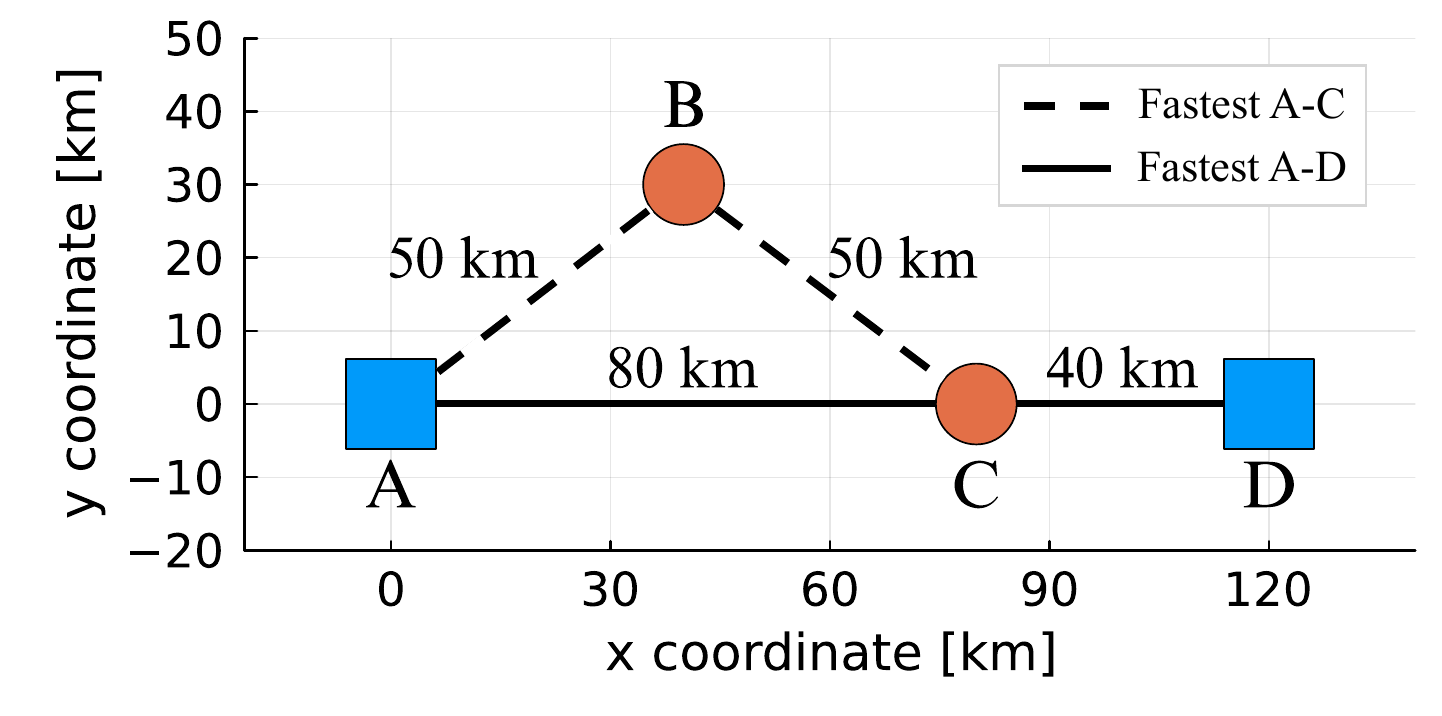}
\caption{
    Example of non-isotonic behavior in quantum-repeater networks.
    The utility function, the secret-key rate $\text{SKR}_{F, T}$ for link fidelity $F=0.94$ and coherence time $T=10$ s, is approximately 32 Hz for path $p_1 = (\text A, \text C)$ and 59 Hz for path $p_2 = (\text A, \text B, \text C)$, making $p_2$  the better path between A and C.
    However, when we append edge $e = (\text C, \text D)$ to both paths, we find approximately 13 Hz for path $p_1 \oplus e$ but exactly 0 for path $p_2 \oplus e$, making $p_1 \oplus e$ the better path between nodes A and D.
    The reason why path $p_2 \oplus e$ has zero utility is because with three links that each have a fidelity as low as $F=0.94$, the error rates in the end-to-end entangled states exceed the tolerance of the BBM92 protocol.
    \label{fig:skr_not_isotonic}
}
\end{figure}
In this example the non-isotonicity arises due to the tradeoff between state quality and entangling rate inherent in the definition of SKR.
One can either make high-quality entangled states slowly or low-quality entangled states quickly to achieve the same SKR.
Splitting the end-to-end distance into many segments can be favorable for the entangling rate as it allows suppression of the losses that scale exponentially with the channel length. In fact, to optimize the entangling rate over a fixed distance in this model, one would introduce as many repeaters as possible to split the total length into as many segments as possible (assuming losses only stem from fiber attenuation).
However, as each repeater is a noisy quantum device and entanglement generation is a noisy process (as captured by the fidelity parameter $F$), the end-to-end state quality will decrease when there are too many repeaters and hence segments.
In the example in Fig.~\ref{fig:skr_not_isotonic}, using two segments between A and C (by visiting node B) achieves a higher entangling rate, which leads to a higher secret-key rate.
However, when there are three segments between A and D, the noise levels are so high that it becomes impossible to generate a secret key, no matter how many pairs are distributed (i.e.~the secret-key fraction is zero).
It is then required to eliminate one segment to keep the state quality high enough, even if it comes at the cost of a reduced entangling rate.
The best path from A to D hence does not visit node B, and the best path from A to D is not obtained by extending the best path from A to C, demonstrating that the utility is not isontonic according to Def.~\ref{def:isotonicity}.

We now make some remarks on the non-isotonicity of the secret-key rate.
First, while the example in Fig.~\ref{fig:skr_not_isotonic} relies on one of the paths having zero SKR, this is not required to achieve non-isotonic behavior;
we have numerically found examples of non-isotonicity where the associated fidelity was well above the threshold for secret-key generation. 
Second, the tradeoff between entangling rate and state quality is not unique to the repeater type, model or performance metric that we utilize here, and hence we expect non-isotonic behavior to occur more generally in repeater networks.
Third, while our example shows how a rate-fidelity tradeoff can generate non-isotonicity even if the rate and fidelity would themselves be isotonic, it turns out that for swap-ASAP repeater chains the rate and fidelity are also themselves non-isotonic.
Hence, even if the goal is only to find the path with either the highest rate or the highest fidelity, pathfinding algorithms that assume isotonicity are not appropriate.
Proof of their non-isotonicity can be found in the form of numerical counterexamples contained in our repository~
\cite{datarepo}
(see \verb|isotonicity_counterexamples/| for concrete examples).
We note that while the BeFS-HEURISTIC and metaheuristic algorithms are expected to work well for pathfinding with the rate or fidelity as utility function, such an investigation is beyond the scope of this paper.

To better understand how non-isotonic $\text{SKR}_{F, T}$ is, we study the difference between the $\text{SKR}_{F, T}$ of the optimal path and the path found by the extended Dijkstra's algorithm. Since the extended Dijkstra's algorithm would find the optimal path for isotonic utility functions, this difference can be seen as a measure of non-isotonicity.
In Fig.~\ref{fig:how_non_isotonic}, we evaluate the relative difference between the SKRs on a set of 100 randomly generated graphs for different link fidelities and coherence times.
The best path is determined using the BeFS-EXACT algorithm, which we will introduce in Sec.~\ref{sec:deterministic_algs}.

\begin{figure}[t!]
\centering
\includegraphics[width=.99\columnwidth]{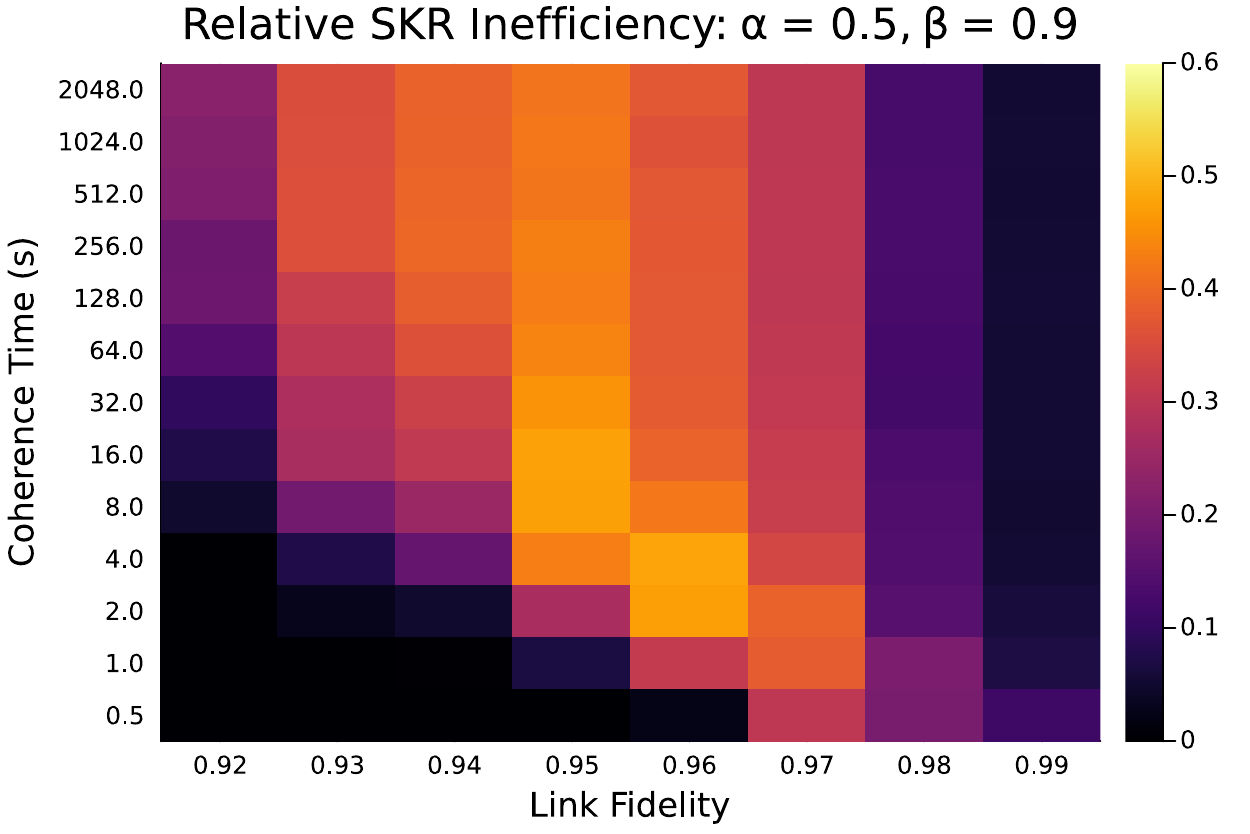}
\caption{Plot of the relative SKR inefficiency (defined as $\left.\left(\text{SKR}_\text{best} - \text{SKR}_\text{dijk}\right)\right/\left(\text{SKR}_{\text{best}}\right)$),
where $\text{SKR}_\text{best}$ is the SKR of the best path and $\text{SKR}_\text{dijk}$ is the SKR of the path found by the extended Dijkstra's algorithm. We plot the relative SKR inefficiency for several link fidelities $F$ and coherence times $T$, where each point was sampled over 100 random Waxman graphs with $\alpha = 0.5$, $\beta = 0.9$, and $L=300\mathrm{km}$. Dijkstra's algorithm always finds the optimal path if $\text{SKR}_{F, T}$ is isotonic, and hence any deviation is evidence of non-isotonicity.
\label{fig:how_non_isotonic}
}
\end{figure}

We generate random graphs, both for Fig.~\ref{fig:how_non_isotonic} and for benchmarking results that we present later in this paper, using the RG1 model developed by Waxman~\cite{waxman1988}.
First, we randomly place repeaters in a $300\mathrm{km} \times 300\mathrm{km}$ square.
Then, we place two end nodes at $\frac1 4$ and $\frac 3 4$ of the way along the diagonal.
Then, for each pair of nodes $u$ and $v$, we generate an edge of length $\|u - v\|_2$ with probability, where $\|\cdot\|_2$ is the Euclidean norm.
$$
P\left(\{u, v\}\right) = \beta \exp\left(-\frac{\|u - v\|_2}{L\alpha}\right),
$$
where $\alpha$, $\beta$, and $L$ are arbitrary parameters.
This procedure leads to graphs where local connections are more likely than more distant connections, roughly approximating the structure of real-life communications networks.
In this paper, we use $\alpha=0.5, \beta=0.9, L=300\mathrm{km}$ because preliminary testing found that it leads to stronger non-isotonicity.

As shown in Fig.~\ref{fig:how_non_isotonic}, non-isotonic behavior appears to peak when the link fidelity is roughly 0.96 and the coherence time is 4 seconds.
The non-isotonicity also appears to approach a non-zero limit as coherence time approaches infinity when link fidelity is around 0.96.
This suggests that the effect of imperfect link entanglements is able to cause non-isotonicity, independent of memory decoherence effects.

\section{Best-First Algorithms}
\label{sec:deterministic_algs}
In this section we introduce a set of algorithms based on the more general principle of best-first search.
The first of these algorithms, called BeFS-EXACT, provably finds the best path and performs significantly better than enumeration.
The second of these algorithms, called BeFS-HEURISTIC, uses a few approximations to achieve significant further performance gains, at the cost of no longer being guaranteed to find the best path (although in practice finding close-to-optimal paths, see Sec.~\ref{sec:algorithm_comparisons}).
Then, we go into detail on the implementation of these approximations.

\subsection{General Best-First Search}

To understand the idea behind a best-first search for finding the best path from node $s$ to node $t$ through graph $G$, let us first introduce some terminology.
Let a \textit{path prefix} be a path that starts at $s$ but ends at a node different than $t$, while a \textit{full path} is a path from $s$ to $t$.
Any path prefix can be \textit{expanded} to create new paths, which means appending edges adjacent to the prefix's end node to the prefix.
Now, let the \textit{path tree} corresponding to $G$, $s$ and $t$ be the rooted tree graph in which each vertex represents either a prefix or a full path such that the root is the empty path at $s$ ($p_\text{empty} = (s)$), the children of each prefix are the paths obtained by expanding it, and the full paths are leaf nodes.
Such a path tree is visualized in Fig.~\ref{fig:path_expansion_example}.
\begin{figure}[t!]
\centering
\includegraphics[width=.99\columnwidth]
{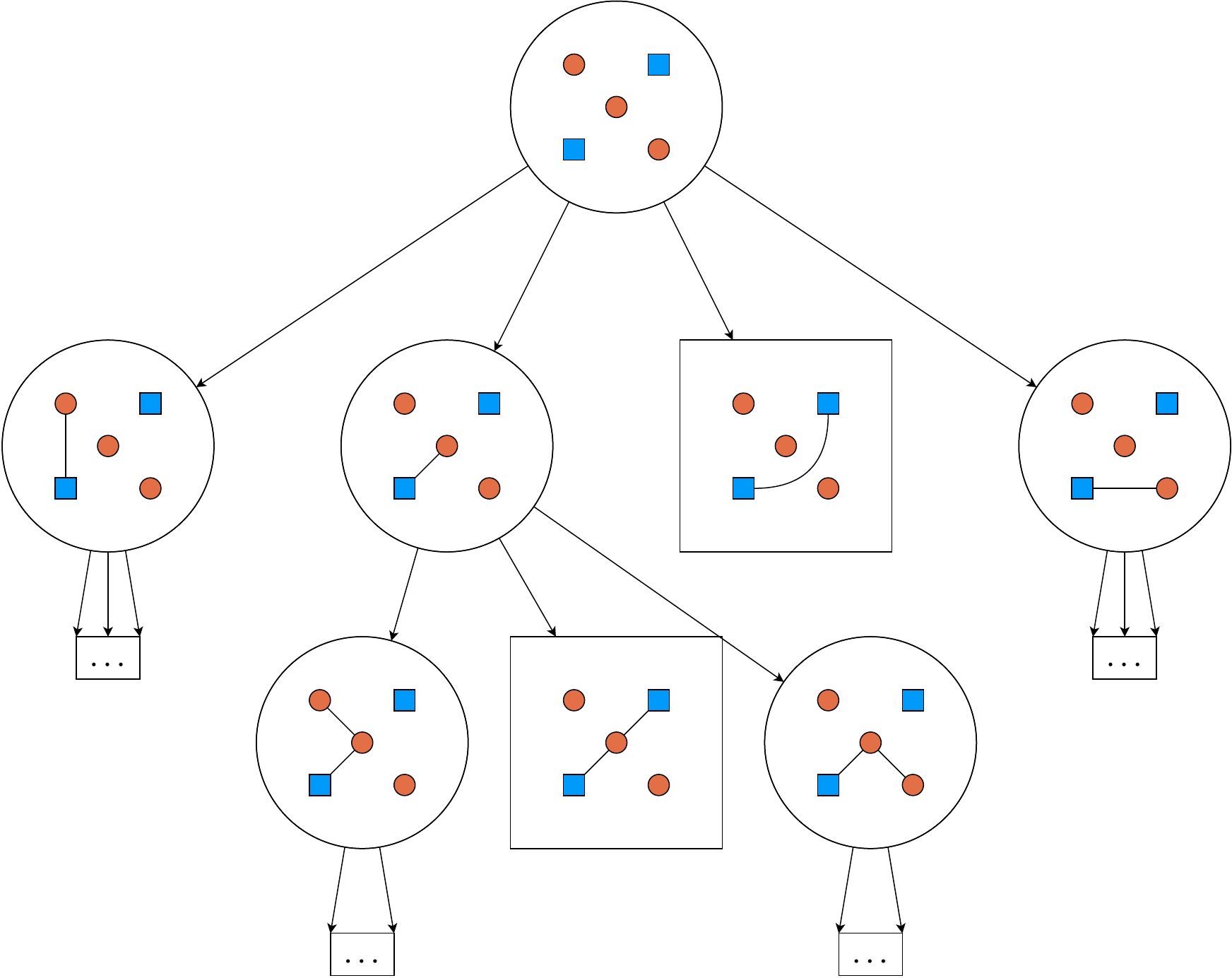}
\caption{
    The path tree for a fully-connected graph with 5 nodes starting at the bottom left blue node.
    The circles are non-leaf nodes, and the squares are leaf nodes.
    The path tree of a pathfinding problem is a tree structure imposed on the set of all paths starting at the start node, and in this diagram, each path is depicted inside of a node in the tree.
    Notice that parent paths and child paths differ by a single edge.
    \label{fig:path_expansion_example}
}
\end{figure}

A best-first search approaches pathfinding as a search through the path tree, where prefixes are expanded to explore the tree graph until the best full path is found.
There are several ways this could be done.
First, it would be possible to expand all path prefixes until all full paths are found.
The utility function can then be evaluated for each of them and the best one can be selected; this is more or less the algorithm presented by Caleffi~\cite{caleffi2017}.
However, as the path tree is very large (the number of full paths scales as $O(n!)$, where $n$ is the total number of nodes), it would be more efficient if only a subset of the path tree could be explored.

Another strategy is to evaluate the utility function for each explored path, and then expand prefixes in a highest-utility-first order.
Assuming that the utility function is monotonic (which it typically should be, even in the quantum case), the paths obtained by expanding a prefix cannot have a higher utility than that prefix itself.
Hence, as soon as a full path is found with a utility that exceeds the utility of any of the currently explored prefixes, it is guaranteed that it is the best full path in the graph and the search can be concluded.
If we associate a positive weight with each edge representing the difference in utility between the two paths, this algorithm can be thought of as Dijkstra's algorithm with multiple destination nodes applied to the path tree (we have transformed the problem into a new pathfinding problem with a simple additive utility function on a much larger graph).
While it can avoid exploring the full graph, it comes at the cost of evaluating the utility function for path prefixes, which is not required when exploring the full graph.

General best-first search generalizes the idea of running Dijkstra's algorithm on the path tree, but instead of evaluating the utility function for the prefixes, instead some \textit{merit function} is calculated and used to decide which path prefix to expand on a best-merit-first basis.
We formally present best-first search in Algorithm~\ref{alg:monotonic}. 
While using the utility function as the merit function will result in the best path being found, there may be other merit functions that guide the search to the correct leaf node of the path tree more quickly.
The key insight here is that the utility function is a destination-agnostic merit function.
Prefixes with equal utility are given equal priority in the search, even if one prefix has an end node that is physically very close to $t$ and the other prefix has an end node that is further from $t$ than $s$ was in the first place.
By using a destination-aware merit function, the most promising prefixes can be explored first, leading to a faster search.
In order to converge to the optimal path, it is required that the merit function on a prefix is an upper bound on the utility of any full path that can be obtained by expanding that prefix (which the utility function satisfies because it is monotonic).
Borrowing from the terminology surrounding the A* algorithm~\cite{hart1968}, we call any such merit function \textit{admissible}.
Generally, the tighter the bound is that the merit function places on the utility of full paths, the fewer path prefixes need to be expanded.

We note that A* is a special case of best-first search, where the merit function is the sum of the utility function and a second function (sometimes referred to as the heuristic function).
Below, we present two different merit functions for swap-ASAP repeater chains to be used in best-first search, effectively giving rise to two different pathfinding algorithms.
As neither merit function is explicitly written as a sum of a merit function and heuristic function, these algorithms are technically instances of best-first search but not of A*.

Since the entanglement-generation protocol is symmetric, the optimal path is preserved when interchanging the source and destination node.
The pathfinding algorithms, however, do distinguish between the source and destination nodes.
This means that there is often a more efficient direction to run each algorithm for any particular graph.
Using a low-cost algorithm at the start to predict this optimal direction can lower average computational costs.
We discuss these algorithms in Appendix~\ref{app:directionality}.

In addition, we avoid exploring dead-ends by removing all of the edges that don't appear on a full path from $s$ to $t$.
We discuss this in Appendix~\ref{app:bcc_pruning}

\begin{algorithm}
\caption{Best-First Search (BeFS)}\label{alg:monotonic}
\begin{algorithmic}

\Require{$G$ is the graph that we pathfind on, $p$ is a path prefix}
\Function{expand}{$G, p$}
  \State{$p_+ := \{\}$}
  \For{$v$ in $\mathrm{neighbors}(\mathrm{last}(p))$}
    \If{$v \not \in p$}
      \State{$p_+ \gets p \oplus (\mathrm{last}(p), v)$}
    \EndIf
  \EndFor
  \State{\Return{$p_+$}}
\EndFunction

\Require{$G$ is the graph that we pathfind on, $s$ is the source node, $t$ is the end node}
\Function{bestFirstSearch}{$G, s, t$}
  \State{$q := \mathrm{PriorityQueue}()$}
  \State{$q \gets ([s], \mathrm{merit}([s]))$}
  \While{$q$ is not empty}
    \State{$p := q.\mathrm{pop}()$}
    \If{$\mathrm{last}(p) = t$}
      \State{\Return{$p$}}
    \EndIf
    \For{$p'$ in $\mathrm{expand}(G, p)$}
      \State{$q \gets (p', \mathrm{merit}(p'))$}
    \EndFor
  \EndWhile
  \Comment{This is unreachable if a path from $s$ to $t$ exists}
\EndFunction

\end{algorithmic}
\end{algorithm}

\subsection{Exact Best-First Search}
\providecommand\dist{{\mathrm{d}}}
\providecommand\skr{{\mathrm{SKR}_{F, T}}}
\providecommand\psetreal{{\mathcal{P}}}
\providecommand\pset{{\mathbb{P}}}
\providecommand\bound{{\mathrm{BOUND}_{F, T}}}
\providecommand\realSuf{{\psetreal_{G}(\mathrm{end}(p), t)}}
\providecommand\lBound{{\dist_G(\mathrm{end}(p), t)}}
\providecommand\fakeSuf{{\pset_{\geq \lBound}}}
\providecommand\fakeSufEx{{\pset_{=\lBound}}}
\providecommand\meritEx{{\mathrm{M}_{\mathrm{e}, F, T}}}
\providecommand\hompath[1]{{H_{\dist_{G}(\mathrm{end}(p), t), {#1}}}}

In this section we introduce our only algorithm that always guarantees that the best possible path is found, i.e., it is our only ``exact'' algorithm, and hence we name it BeFS-EXACT.
The algorithm consists of best-first search (Algorithm~\ref{alg:monotonic}) with a specific merit function, which we will derive as follows.
To begin, we start by attempting to find the "best" admissible merit function.
Since tighter bounds increase the performance of BeFS, this should be the tight bound where the merit of each path is equal to the utility of the best path extension.
If we let $\psetreal_G(u, v)$ be the set of paths from $u$ to $v$ in $G$, we can write this formally as:

\begin{equation}
\mathrm{M}_{1, F, T}(p) = \max_{p' \in \realSuf} \skr(p \oplus p')
\end{equation}
In other words, the best admissible merit function maps $p$ to the greatest utility attainable by extending $p$ to form a full path.
While this merit function is admissible, it seems unlikely that it can be solved without using an algorithm that solves for the actual path extension, which would make this algorithm redundant.
Since this tight merit function is unlikely to be useful in efficiently finding paths, another merit function must be chosen.

Since appending elements to a set cannot decrease the maximum of that set, extending $\realSuf$ with more paths cannot make the resulting merit function inadmissible.
Thus, one way to modify this merit function is by constructing a larger set of paths that admits a simple optimization procedure.
If we let $\dist_G$ be the graph distance function for $G$ (the minimum physical distance along the network), we can define a new set of paths, $\fakeSuf$, to be the set of all conceivable paths whose length is at least $\lBound$.

To provide some examples of the paths included in $\fakeSuf$, $\fakeSuf$ includes every path in $\realSuf$, since each one has a physical length of at least $\lBound$ by definition.
It also includes any path that is not in $G$, with any number of segments and any sequence of lengths where the total length is at least $\lBound$.

Plugging this set into our merit function gives
\begin{equation}
\label{eq:befs_exact_version_1}
\mathrm{M}_{2, F, T} = \max_{p' \in \fakeSuf} \skr(p \oplus p')
\end{equation}

Since $\fakeSuf$ is defined with fewer constraints, we might hope to find a simple, analytical formula for solving Eq.~\eqref{eq:befs_exact_version_1}.
Unfortunately, memory decoherence affects the secret-key rate in a way that depends on the configuration of the quantum repeaters in a rather complicated way, making it difficult to analytically derive the suffix $p'$ for which $\text{SKR}_{F, T}(p \oplus p')$ is maximized.

To remedy this, we replace $\skr(p \oplus p')$ with $\bound(p, p')$, where $\bound(p, p')$ is the secret key rate of $p \oplus p'$ if every node adjacent to $p'$ has perfect quantum memory.
This must be greater than or equal to $\skr(p \oplus p')$ because removing memory decoherence improves or preserves the fidelity, which improves or preserves the secret-key rate.
Since none of the individual secret-key rates are decreased, the maximum secret-key rate cannot be decreased, and admissibility is preserved.
We define the merit function of BeFS-EXACT for any path prefix $p$ that ends at node $\mathrm{end}(p)$ as
\begin{equation}
\meritEx(p) = \max_{p' \in \fakeSuf} \bound(p, p').
\end{equation}

We will now summarize the derivation of an efficient formula for $\meritEx(p)$ based on theoretical results about swap-ASAP repeater chains that we have included in Appendix~\ref{app:bounds}.

Firstly, as shown in Corollary~\ref{cor:skr_of_extension_best_for_smallest_length}, any suffix $p'$ whose total length is greater than $\lBound$ can improve its entanglement generation time by shortening one or more edges to get a total length of $\lBound$.
Secondly, following from Lemma~\ref{lem:skr_optimal_for_hom_extension}, the path $p': \mathrm{dist}(p) = \lBound$ that maximizes $\bound(p, p')$ is always a \textit{homogeneous} one, where each edge in the chain has the same length.
If we let $H_{L, N}$ be a path consisting of $N$ edges whose lengths are all equal to $\frac{L}{N}$, we then have
\begin{equation}
\meritEx(p) = \max_{N \in \mathbb N} \bound(p, \hompath{N}).
\end{equation}

Now, as each state is generated with a fidelity of $F$ (and the memory noise is only affected by $p$), the noise levels will increase monotonically with $N$ as long as $F < 1$ (for $F=1$, this algorithm is not applicable).
Hence, there will be a finite value $N^*$ such that $\text{SKF} = 0$ (and hence $\text{SKR} = 0$) whenever the total number of edges in the combined paths is equal to or larger than $N^*$.
It is implied by Lemma~\ref{lem:max_N} that this is true for
\begin{equation}
N^* = \frac{ \log \left( 1 - 2 \cdot Q^* \right) } {\log \left( \frac {4F - 1}{3} \right) },
\end{equation}
where $Q^* \approx 0.110028$ is a constant that follows from the fidelity threshold for secret-key generation using BBM92.

Since the edges in the prefix also count towards this bound, the merit function can therefore be evaluated by performing a brute-force search over the domain $D = 1, 2, ..., \lfloor N^* - \left|\mathrm{edges}(p)\right|\rfloor$, giving us the formula
\begin{equation}
\meritEx(p) = \max_{N \in D} \bound(p, \hompath{N}).
\end{equation}
Since $\bound$ is evaluated numerically using QuantumNetworkRecipes.jl, this formula can be directly and efficiently evaluated.

\subsection{Heuristic Best-First Search}
\label{sec:heuristic_befs}

In this section we introduce BeFS-HEURISTIC, our second algorithm based on best-first search.
Unlike BeFS-EXACT, it is not an exact algorithm because it uses a few heuristics to speed up computation, which means that the found path is no longer guaranteed to be optimal.
Importantly, these heuristics allow BeFS-HEURISTIC to find paths more quickly than BeFS-EXACT, and which algorithm is preferable depends on the tolerance for path imperfection and the available computational resources.
In practice, BeFS-HEURISTIC is able to find the optimal path (or a path that is close to optimal) most of the time, as we will show in Sec.~\ref{sec:algorithm_comparisons}.

While the merit function of BeFS-EXACT is provably an upper bound on the utility of the best path that can be constructed using a given path prefix, it requires the assumption that  repeaters adjacent to the new edges do not suffer memory noise.
This assumption allows us to prove that the best path extension is always a homogeneous one (i.e., one in which all edges are equally long), which in turn allows us to significantly reduce the search space.
However, removing memory noise might lead to the utility of full paths being significantly overestimated, resulting in a merit function that only loosely bounds the utility of full paths.
As tighter bounds reduce the cost of best-first search, the main idea behind BeFS-HEURISTIC is to improve the merit function of BeFS-EXACT by including memory noise, at the cost of losing exactness.
Then, for any path prefix $p$ that ends at $r$ we have the merit function
\begin{equation}
\mathrm{M}_{\text{HEURISTIC}, F, T}(p) = \max_{N \in D} \text{SKR}_{F, T}(p \oplus H_{L_{G, r, t}, N, r, t}).
\end{equation}
As a homogeneous path extension is not generally the optimal path extension, this is not a valid upper bound and hence the merit function is technically not admissible.
In practice, however, this is not an issue.
Moreover, BeFS-HEURISTIC uses two further heuristics to reduce computational overhead.
$\text{SKR}_{F, T}(p \oplus H_{L_{G, r, t}, N, r, t})$ typically has a unique local maximum, which allows us to efficiently find the maximum by performing a binary search over $N$ as described in Appendix~\ref{app:binary_search}.
In addition, a cheap-to-calculate partial-ordering relation between paths can be used to defer (and often, avoid) merit-function evaluations.
The majority of time spent while running each of the pathfinding algorithms goes to evaluating the merit function.
In BeFS-HEURISTIC we employ a heuristic that is based on domination relations to reduce these costs.
We say that path $p$ dominates path $p'$ if $p$ and $p'$ are not identical and there is an order-preserving injection from the edge lengths in $p$ into the edge lengths of $p'$ where each edge in $p$ is at least as short as the corresponding edge in $p'$.
Said differently, $p$ dominates $p'$ if $p$ can be created from $p'$ by first removing edges and then reducing the lengths of the remaining edges.
Generally, if a path $p$ dominates $p'$, then $p$ is very likely to have a higher utility than $p'$.
However, in repeater chains with memory decoherence, this is not always true.
For instance, a homogeneous three-link chain can outperform the same chain but with the middle edge made very short as the discrepancy in entanglement-generation times between neighboring edges will result in large qubit storage times, even though the homogeneous chain dominates the second chain.
If we additionally find that the end of $p$ has a lower graph distance to the destination than the end of $p'$, then $p$ is also very likely to have a greater merit than $p'$.
These relations provide a method for proving that a path prefix is unlikely to be at the head of the queue without explicitly evaluating the merit function.
This allows us to defer merit evaluation on a given prefix until it is no longer dominated by anything else in the queue.
We call this optimization \textit{prefix bounding}.
Note that we simply replace the merit function evaluation with a domination check for some of the elements of the queue.
Indeed, if we assume that the dominance relation is always correct, we can prove that for any given graph, the set of elements in the queue at any given step is not changed by the addition of prefix bounding.
This means that this algorithm is likely just a constant factor more efficient, rather than being asymptotically more efficient.
However, as shown in Appendix~\ref{app:complexity validity}, the cost of dominance comparisons is negligible compared to the cost of SKR estimations, resulting in an exponential speed up in practical terms as shown in Sec.~\ref{sec:algorithm_comparisons}.
(Within the tested regime, we have observed that prefix bounding reduces the algorithm's query count from exponential to sublinear in the network size.)
Prefix bounding is implemented by replacing the priority queue in Algorithm~\ref{alg:monotonic} with a priority queue paired with a partially ordered set, as described in detail in Algorithm~\ref{alg:prefix_bounding}.

\begin{algorithm}
\caption{Best-first search with prefix bounding}\label{alg:prefix_bounding}
\begin{algorithmic}
\Require{$G$ is the graph that we pathfind on, $s$ is the source node, $t$ is the end node}
\Function{prefixBoundingSearch}{$G, s, t$}
  \State{$q := \mathrm{PriorityQueue}()$}
  \State{$q \leftarrow ([s], \mathrm{M}([s]))$}
  \State{$b := \text{a partially ordered set initialized as }\{[s]\}$}

  \While{$\lnot \mathrm{empty}(q)$}
    \State{$p, m := q.\mathrm{pop}()$}
    \State{$b := b \setminus \{p\}$}

    \If{$\mathrm{end}(p) = t$}
      \State\Return{$p$}
    \EndIf

    \For{$p'$ in $\mathrm{expand}(G,p)$}
      \State{$b := b \cup \{p'\}$}
      \For{$p''$ in $b$} 
        \If{$p'$ dominates $p''$}
          \State{add $p'' \leq p'$ to $b$}
        \EndIf
        \If{$p''$ dominates $p'$}
          \State{add $p' \leq p''$ to $b$}
        \EndIf
      \EndFor
    \EndFor

    \For{all $r \not \in q$ that are maximal within $b$}
      \State{$q[r] = \mathrm{merit}(r)$}
    \EndFor
  \EndWhile
\EndFunction
\end{algorithmic}
\end{algorithm}

\section{Metaheuristic algorithms}
\label{sec:metaheuristic_algs}

In this section we introduce a second family of pathfinding algorithms for dealing with the non-isotonic behavior of $\skr$.
We present two metaheuristic methods, one based on simulated annealing and another on genetic algorithms.
These methods are attractive because they do not use model-specific assumptions, making them generally applicable to any quantum-repeater model, or indeed any non-isotonic utility function.
The price to pay for this is that they might be less effective at finding good paths for the exact model under consideration here than our more tailored algorithms, as we will see in Sec.~\ref{sec:algorithm_comparisons}.
A second attractive property about these algorithms is the ability to tune a tradeoff between effectiveness and computational cost, making it possible to control the overhead even for large networks by tolerating an appropriate degree of suboptimality.

\subsection{Simulated Annealing}
Simulated annealing is a metaheuristic algorithm inspired by thermodynamics~\cite{russellnorvig2010}.
At each time step $i$, the algorithm has a current path $p$ from $s$ to $t$ and a temperature $\theta_i > 0$ that decreases as a function of $i$.
Exactly how $\theta_i$ decreases with $i$ is a parameter of the algorithm called the \textit{cooling schedule}.
Each step, we randomly mutate $p$ to get a slightly different path $p'$, as explained in the next paragraph.
We then measure the utilities of both $p$ and $p'$ and compare them.
With a certain probability depending on the difference $\Delta$ in utilities and the current temperature $\theta_i$, we accept $p'$ and set $p \gets p'$; otherwise we leave $p$ unchanged.
Near the beginning when the temperature $\theta_i$ is large, the algorithm behaves like a random walk through the space of all paths from $s$ to $t$.
Near the end when $\theta_i$ is small, it behaves more like a greedy hill-climbing algorithm.
It is important that the algorithm occasionally accepts candidate paths $p'$ that are worse than the current path $p$, as this helps it escape local maxima in the utility function landscape.

To mutate a path $p$, we pick uniformly at random from the set of all possible mutations of $p$.
A path is considered to be a mutation of $p$ if it can be obtained from $p$ by one of the following two operations:
\begin{enumerate}
\item 
Add a repeater: Replace an edge $(a, c)$ in $p$ with the edges $(a, b)$ and $(b, c)$.
We only allow this operation if the repeater $b$ is not already in $p$ and the edges $(a, b)$ and $(b, c)$ actually exist in the graph $G$.

\item 
Remove a repeater: Replace two consecutive edges $(a, b)$ and $(b, c)$ in $p$ with the edge $(a, c)$.
We only allow this operation if the edge $(a, c)$ actually exists in the graph $G$.
\end{enumerate}
These mutations are illustrated in Fig.~\ref{fig:mutation}.

\begin{figure}[t!]
\centering
\includegraphics[width=.99\columnwidth]{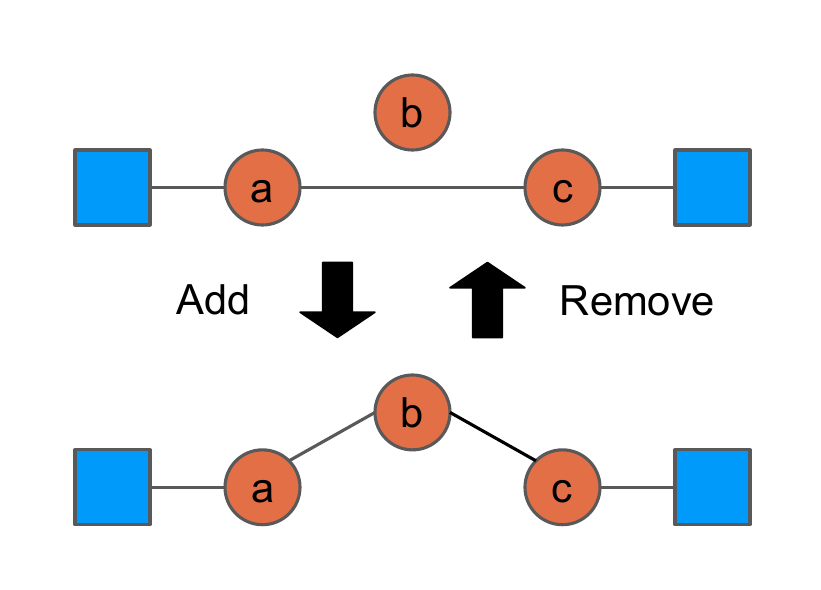}
\caption{
    Examples of the two types of path mutations used in the simulated annealing and genetic algorithms.
    \label{fig:mutation}
}
\end{figure}

We also experimented with a third mutation operation called \textit{crossing / uncrossing edges}, in which one replaces a pair of edges $(a, b)$ and $(c, d)$ in $p$ with the new edges $(a, c)$ and $(b, d)$, effectively reversing the entire subpath of nodes between $b$ and $c$.
We found this operation to give no major benefit to either simulated annealing or the genetic algorithm described below, so for simplicity we just discarded it. 

One issue we found is that when a path has too many edges, it has a SKR of zero.
This means that in the high-dimensional landscape of paths from $s$ to $t$ connected by mutations, there are wide flat valleys of paths with many edges with zero utility, and all of those candidate neighbors also have zero utility.
This is a problem for simulated annealing, which would often get stuck in such valleys.
We fixed this by modifying the utility function on paths; if the current path $p$ has a SKR of zero, then we penalize it further based on the number of edges.
Specifically, we define:

\newcommand{\skrbar}{\overline{\text{SKR}}_{F,T}}

\[ \skrbar(p) := \begin{cases}
    \skr(p) & \skr(p) > 0 \\
    -\ell \cdot (\text{\# of edges of } p) & \skr(p) = 0 \\    
\end{cases} \]

Here $\ell \ge 0$ is an algorithm parameter which we call the \textit{length penalty}. The penalty term encourages the algorithm to remove repeaters from $p$ when it is too long, allowing it to find its way back to the region of paths with positive SKR.

The pseudocode is given in Algorithm \ref{alg:annealing}.

\begin{algorithm}
  \caption{Simulated annealing}
  \label{alg:annealing}
  \begin{algorithmic}
    \State $p \gets$ initial path from $s$ to $t$
    \For {$i$ from 1 to $n$}
      \State $p' \gets$ randomly-mutate($p$)
      \State $\Delta \gets \skrbar(p') - \skrbar(p)$
      \State With probability $\min(1, \exp(\Delta / \theta_i))$:
        \State \hspace{1em} $p \gets p'$
    \EndFor
    \State \textbf{return} $p$
  \end{algorithmic}
\end{algorithm}

\subsection{Genetic Algorithm}
Short subpaths of a good path are often useful building blocks on their own.
As such, one can try to make a good path by concatenating subpaths of other good paths. 
Algorithms using this strategy are called genetic algorithms, inspired by genetic recombination in biology~\cite{russellnorvig2010}.

See Algorithm~\ref{alg:genetic} for the pseudocode.
At each generation $i$, we generate a population $P_i$ of $m$ candidate paths from $s$ to $t$.
To do this, we repeatedly select parent paths $p_1, p_2 \in P_{i-1}$ (possibly the same path twice) and recombine them to make a child path $p_c$, with a chance of an extra mutation (as explained below).
To bias parent selection toward better paths while still maintaining population diversity, we select paths randomly with probability given by the softmax function~\cite{goodfellow2016}. 
More explicitly:
\[ \text{Pr}(p\text{ is selected}) = \frac{\exp(\skr(p) / \theta)}{\sum_{p' \in P_i} \exp(\skr(p') / \theta)} \]
Here $\theta > 0$ is a parameter of the algorithm called the \textit{selection temperature} and plays a similar role here as in the simulated annealing algorithm.
When $\theta \approx 0$ the algorithm selects the best path nearly every time, which decreases diversity.
On the other hand when $\theta \to \infty$ it selects parents nearly uniformly at random.
Thus the temperature $\theta$ represents the tradeoff between the desire to increase the average fitness of the population (exploitation) and the desire to keep the population diverse (exploration).

\begin{algorithm}
  \caption{Genetic algorithm}
  \label{alg:genetic}
  \begin{algorithmic}
    \State $P_0 \gets$ initial population of paths from $s$ to $t$
    \For {$i$ from 1 to $n$}
      \State $P_i \gets \emptyset$
      \For {$j$ from 1 to $m$}
        \State $p_1 \gets$ select-parent($P_{i-1}, \theta$)
        \State $p_2 \gets$ select-parent($P_{i-1}, \theta$)
        \State $p_c \gets$ recombine($p_1, p_2$)
        \State With probability $r$:
        \State \hspace{1em} $p_c \gets$ randomly-mutate($p_c$)
        \State $P_i \gets P_i \cup \{p_c\}$
      \EndFor
    \EndFor
    \State \textbf{return} $\argmax_{p \in P_n} \skr(p)$
  \end{algorithmic}
\end{algorithm}

To recombine two parents $p_1$ and $p_2$, we pick uniformly at random from the set of all possible valid recombinations of $p_1$ and $p_2$.
We define a \textit{valid recombination} of $p_1$ and $p_2$ to be a choice of a prefix $s \rightsquigarrow a$ of $p_1$'s repeaters and a suffix $b \rightsquigarrow t$ of $p_2$'s repeaters such that the edge $(a, b)$ exists in the graph $G$.
To apply the recombination, we simply concatenate the prefix and suffix to obtain a new path $s \rightsquigarrow a \to b \rightsquigarrow t$ from $s$ to $t$.
If the prefix of $p_1$ and the suffix of $p_2$ share some repeaters, then the concatenation will have duplicate repeaters, creating a loop from a repeater to itself.
If the concatenation has such a loop, then we just remove it, both to strictly improve the child's utility and to simplify the mutation algorithm.
Finally, we randomly mutate the child path with a probability $r$ called the \textit{mutation rate}, which is another algorithm parameter.
We use the same random mutation function as in the simulated annealing algorithm.
The whole process is illustrated in Fig.~\ref{fig:recombination}. 

\begin{figure}[t!]
\centering
\includegraphics[width=.99\columnwidth]{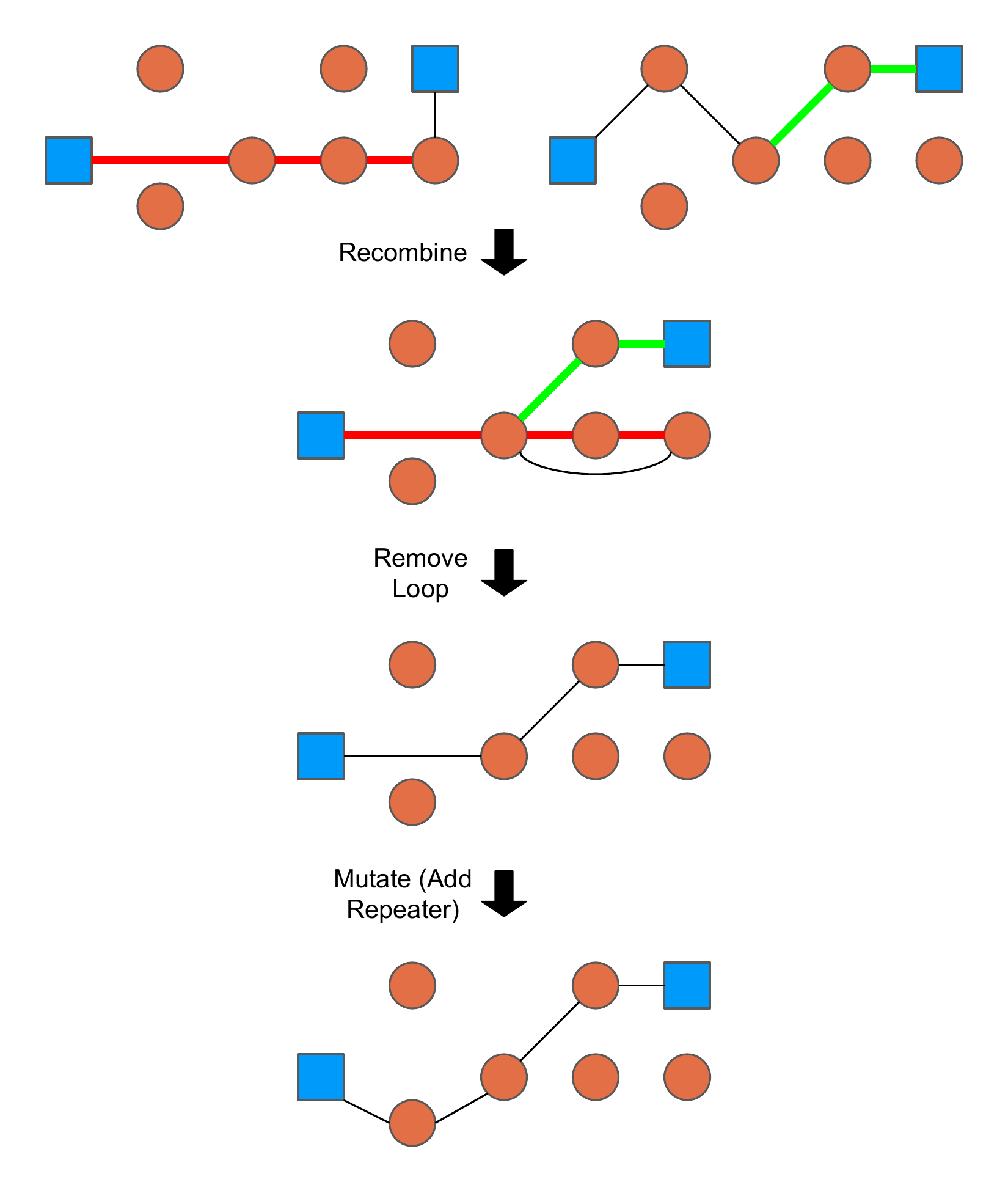}
\caption{
    An example recombination of two parent paths $p_1$ and $p_2$.
    We concatenate the red prefix of $p_1$ with the green suffix of $p_2$, joining them with the black edge.
    This creates a loop of length 3, which we then remove.
    Finally we apply a random mutation, in this case adding a repeater.
    \label{fig:recombination}
}
\end{figure}

\subsection{Selecting Good Parameters}
Both of these metaheuristic algorithms have a number of adjustable parameters.
Namely, simulated annealing has the number of iterations $n$, the length penalty $\ell$, and the cooling schedule dictating the temperature $\theta_i$ at each time step $i$.
The genetic algorithm has the number of generations $n$, the population size $m$, the selection temperature $\theta$, and the mutation rate $r$.
These parameters have a large impact on the algorithms' effectiveness, making it important to choose them carefully.
We now describe a procedure in which we run both algorithms on a set of randomly generated graphs in order to find good parameters.

In choosing our parameter values, due to CPU time constraints, we budgeted 500 calls to the utility function $\skr$ to the algorithms.
In other words, we fixed $n = 500$ in the case of simulated annealing and $n = 500 / m$ in the case of the genetic algorithm.
To tune the other parameters, we let each parameter run through about 5-10 different values, spanning several orders of magnitude when possible.
For example, for the genetic algorithm we let the selection temperature $\theta$ take values in the set $\{ 0, 0.1, 0.2, 0.5, 1, 2, 5, 10, 20, 50, 100 \}$.
For each combination of parameter values, we generated 100 random Waxman graphs for each repeater count in the set $\{ 20, 40, 60, 80, 100 \}$.
We then selected the parameter values resulting in the highest average SKR.

For the cooling schedule $\theta_i$, we tried two different types of cooling functions, namely linear functions $\theta_i = \left.\theta_0 \cdot \left( 1 - \frac{i}{n} \right)\right.$, and exponential functions $\theta_i = \theta_0 \cdot (\theta_f / \theta_0)^{i/n}$.
Here we fixed $\theta_f = 0.1$ and let $\theta_0$ run through the set $\{ 0.1, 0.2, 0.5, 1, 2, 5, 10, 20, 50, 100 \}$. 

For simulated annealing at $n = 500$, the best values found were $\ell = 2$ with a linear cooling schedule $\theta_i = 2 \cdot \left( 1 - \frac{i}{n} \right)$.
For the genetic algorithm, the best values found were $m = 25$ (and thus $n = 500 / 25 = 20$), $\theta = 0.5$, and $r = 0.75$.

\section{Performance of Algorithms}
\label{sec:algorithm_comparisons}
To evaluate the performance of our algorithms, we ran each of the algorithms described above, in addition to the enumeration method developed by Caleffi~\cite{caleffi2017}, on a set of randomly generated graphs, and we note both the query count and the relative SKR inefficiency of each algorithm.

The query count is a metric for how fast each of the algorithms is, and it is defined as the number of calls to the subroutine for estimating the secret key rate.
As shown in Appendix~\ref{app:complexity validity}, this is a roughly linear approximation of the actual CPU time.
However, this linear relationship breaks down for large repeater counts, as shown in Appendix~\ref{app:tradeoff_1000}, likely due to excessive overhead from dominance evaluations.

The relative SKR inefficiency is a metric for the quality of a particular solution to a particular pathfinding algorithm relative to a set of solutions.
The estimated optimal secret-key rate is the lowest secret-key rate among the proposed solutions to the particular pathfinding problem.
The secret-key rate inefficiency is the difference between the estimated optimal secret-key rate and the secret-key rate of the particular path.
The relative secret-key rate is the secret-key rate inefficiency divided by the estimated optimal secret-key rate.
For instance, if the best path from $s$ to $t$ in a network has a SKR of 5 hertz, and the path found in a particular run is 4 hertz, then the relative SKR inefficiency is 20\%, since the SKR that was achieved is 20\% lower than the maximum SKR.
In addition, if there is no usable path through a repeater network (i.e.~the best achievable SKR is 0), then all of the algorithms are correct, and the relative SKR inefficiency is defined as 0.
Since our utility function is evaluated using a Monte-Carlo method (as discussed in Sec.~\ref{sec:utility_function}), there is some inherent error, meaning that even algorithms that should be perfect may occasionally produce incorrect results.
However, as shown in Appendix~\ref{app:monte_carlo_error}, this is not an issue in practice.

\begin{figure}[t!]
\centering
\includegraphics[width=.99\columnwidth]{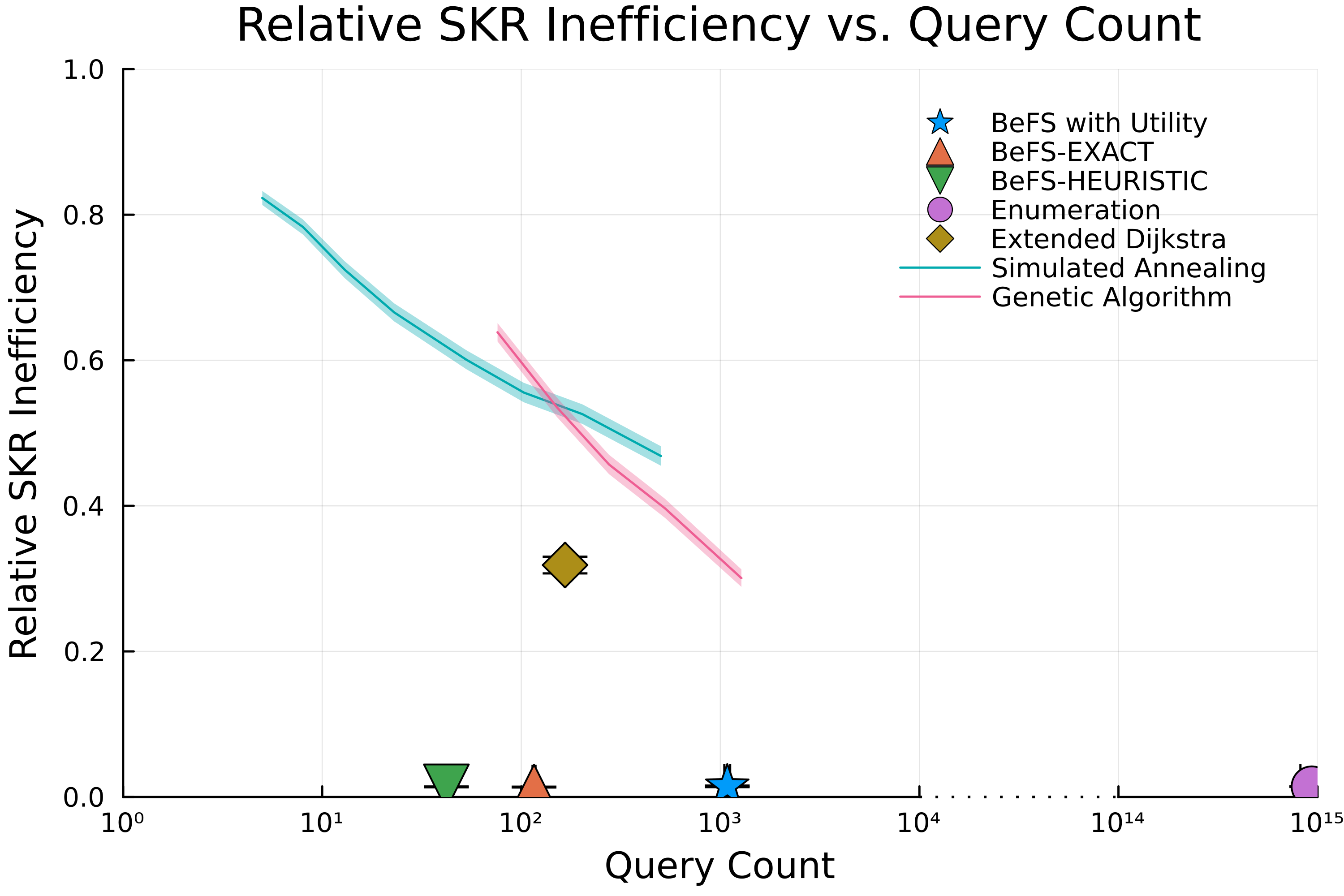}
\caption{
    A scatterplot of the query count of the algorithm (x-axis) vs. the relative secret key rate inefficiency (y-axis) when benchmarking the algorithms described in this paper on a set of randomly generated graphs.
    The query count is the number of times that the SKR estimation algorithm is called.
    The algorithms were run on a set of 1000 randomly generated Waxman graphs that each have 25 repeaters.
    The graphs were generated with the parameters $\alpha=0.5, \beta=0.9, L=300\mathrm{km}$.
    The metaheuristic algorithms (simulated annealing and genetic algorithms) are parameterized by the number of steps, so they are plotted as lines rather than points.
    The simulated annealing algorithms are plotted for 1 to 500 steps, and the genetic algorithms are plotted for 1 to 50 generations.
    \label{fig:query_complexity_vs_skr_deficit}
}
\end{figure}
In Fig.~\ref{fig:query_complexity_vs_skr_deficit}, we run each of the algorithms described above on a set of 1000 random Waxman graphs with 25 repeaters each and plot the query count against the relative SKR inefficiency.
For the BeFS algorithms, the figure indicates that the relative SKR inefficiency is negligible (roughly 0), with the fastest algorithm being BeFS-HEURISTIC.
Since BeFS-HEURISTIC has the lowest query count and the same relative SKR inefficiency, this figure indicates BeFS-HEURISTIC performs the best in practice.
For the simulated annealing and genetic algorithms, the figure shows an inverse correlation between query count and relative SKR inefficiency, suggesting that there is a trade-off between these metrics within both families of algorithms.
For this particular graph configuration, genetic algorithms appear to surpass simulated annealing algorithms in terms of quality past a certain number of repeaters.
The trade-off point appears to be around $120$ queries, and as shown in Appendix~\ref{app:tradeoff_50}.
Furthermore, whilst the extended Dijkstra's algorithm is strictly better (lower relative SKR inefficiency and lower query complexity) than most parameterizations of the metaheuristic algorithms for this graph configuration, Appendix~\ref{app:tradeoff_50} shows that the metaheuristic algorithms overtake the extended Dijkstra's algorithm for higher repeater counts.

\begin{figure}[t!]
\centering
\includegraphics[width=.99\columnwidth]{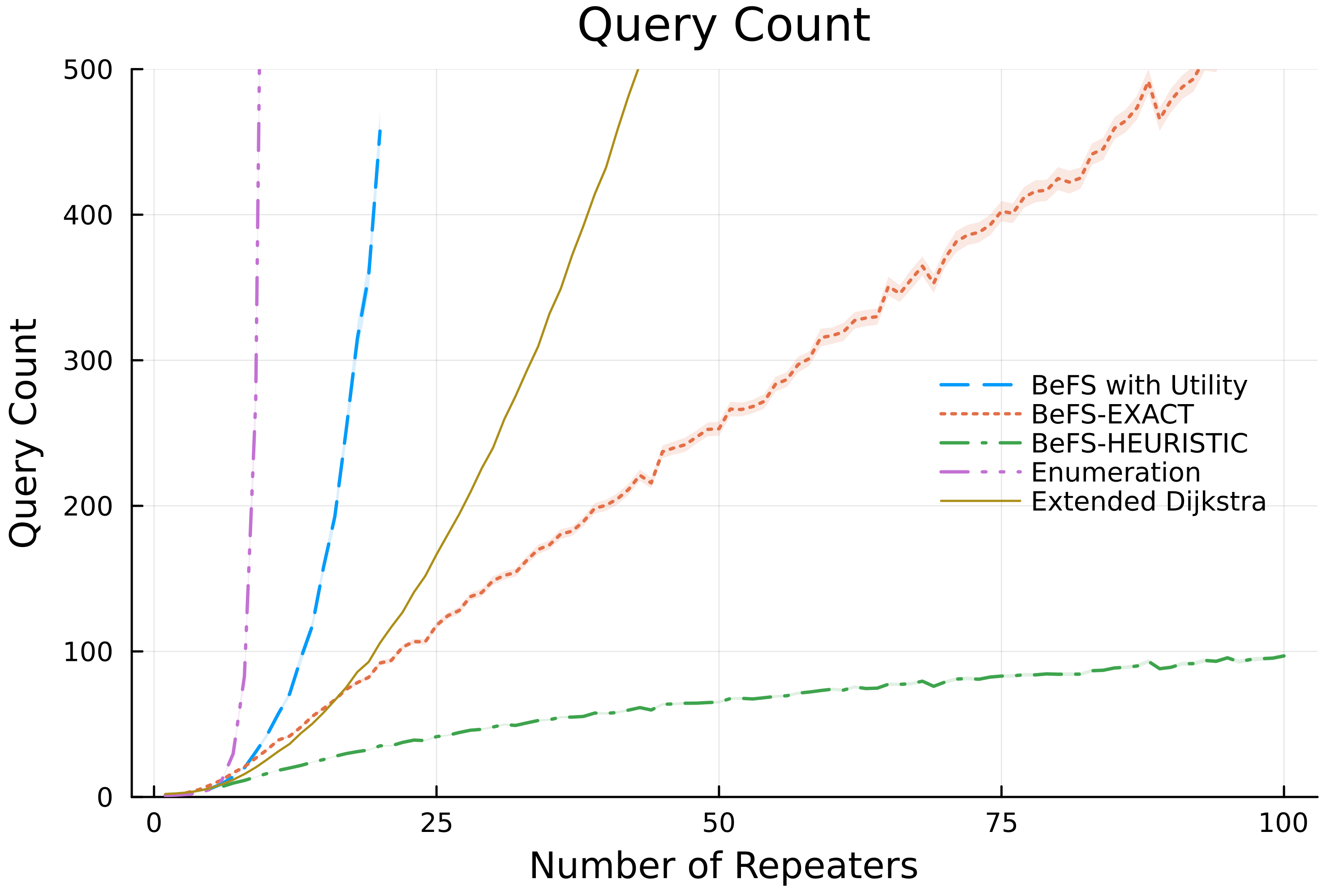}
\caption{
    A plot of the query count of the algorithm (y-axis) vs. the number of repeaters in the graph (x-axis) when benchmarking the algorithms discussed in this paper on a set of randomly generated graphs.
    The query count is the number of times that the SKR estimation algorithm is called.
    The set of graphs was created by generating 1000 random Waxman graphs with $\alpha=0.5, \beta=0.9, L=300\mathrm{km}$ for each graph size between 1 and 100 repeaters (inclusive).
    Only the metaheuristic algorithms have been excluded, since they have a fixed query count, regardless of the graph being evaluated.
    \label{fig:num_rep_vs_query_complexity}
}
\end{figure}
In Fig.~\ref{fig:num_rep_vs_query_complexity}, we plot the query count of each algorithm against the number of repeaters, with 1000 graphs being used for each repeater count.
Since simulated annealing and genetic algorithms have a constant number of function calls, they would form perfectly horizontal lines on this graph and are thus not included.
This plot shows that within the tested parameter space, enumeration appears to have a super-linear, rapidly increasing query count with respect to the number of repeaters, matching the theoretical bounds and becoming intractable for $n \geq 10$.
Surprisingly, the plot shows all of the best-first algorithms having better asymptotic behavior compared to the extended Dijkstra's algorithm, with BeFS-HEURISTIC in particular having sublinear query count in the tested regime.
As described in Sec.~\ref{sec:deterministic_algs}, we suspect that this is due to SKR estimation queries being replaced with cheaper dominance comparisons.

\begin{figure}[t!]
\centering
\includegraphics[width=.99\columnwidth]{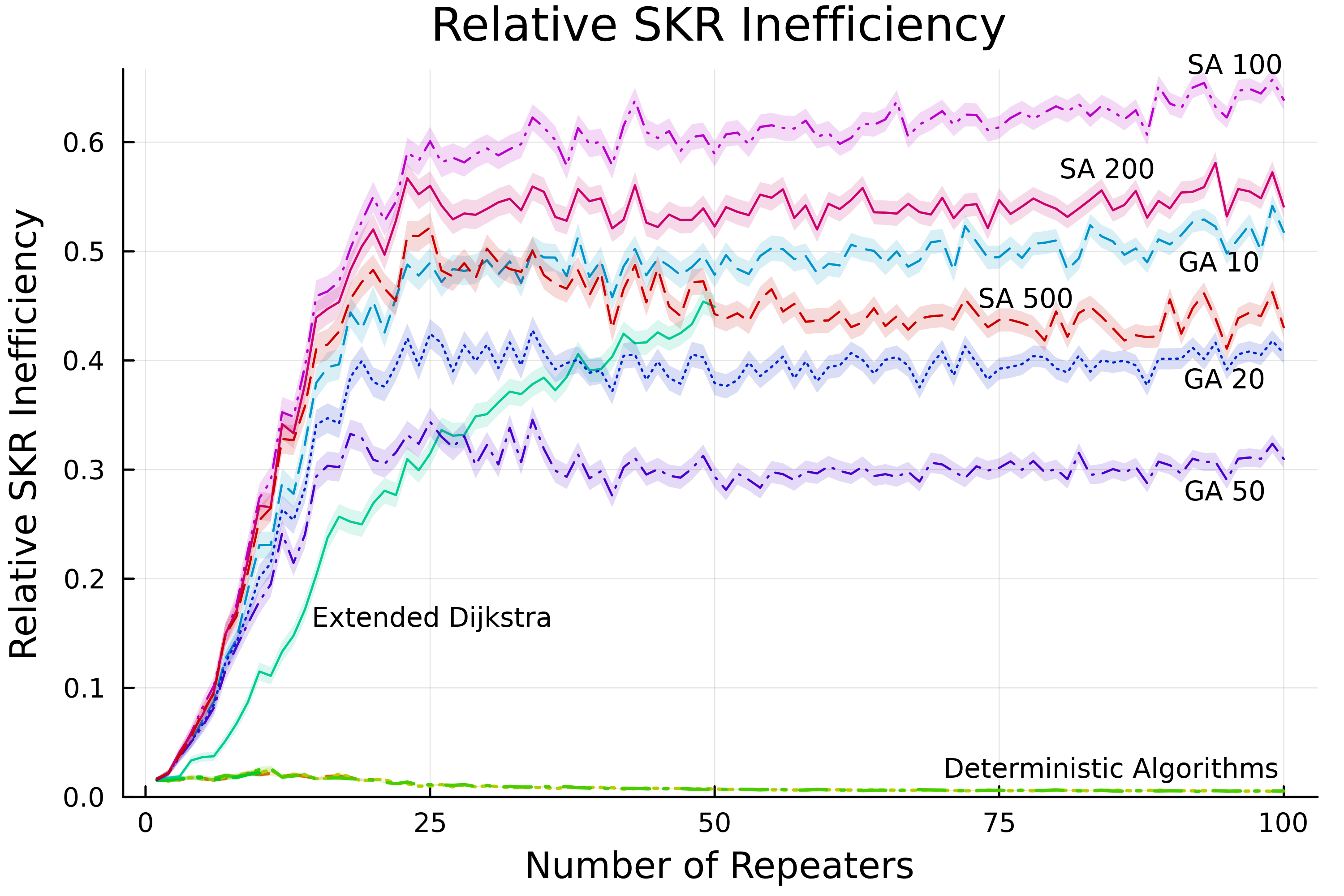}
\caption{
    A plot of the relative SKR inefficiency (y-axis) vs. the number of repeaters in the graph (x-axis) when benchmarking the algorithms discussed in this paper on a set of randomly generated graphs. 
    The query count is the number of times that the SKR estimation algorithm is called.
    The simulated annealing algorithms were evaluated for 100, 200, and 500 steps, and the genetic algorithms were evaluated for 10, 20 and 50 generations.
    The set of graphs was created by generating 1000 random Waxman graphs with $\alpha=0.5, \beta=0.9, L=300\mathrm{km}$ for each graph size between 1 and 100 repeaters (inclusive).
    The algorithms are labeled adjacent to the lines, with simulated annealing indicated by "SA (number of steps)", genetic algorithms indicated by "GA (number of generations)", and "Deterministic Algorithms" being a group consisting of enumeration, BeFS with utility, BeFS-EXACT, and BeFS-HEURISTIC (see Appendix~\ref{app:monte_carlo_error} for more a zoom-in on these algorithms).
    \label{fig:num_rep_vs_skr_deficit}
}
\end{figure}
In Fig.~\ref{fig:num_rep_vs_skr_deficit} we plot the average relative SKR inefficiency of each algorithm against the number of repeaters.
We see that the average relative SKR inefficiency of the exact and near-exact algorithms are close to 0 for all repeater counts, and as shown in Appendix~\ref{app:monte_carlo_error}, the small error explainable by error in the utility function, confirming the hypothesis that these algorithms reliably produce the optimal path.
The extended Dijkstra's algorithm's solutions appear to decrease in relative quality as the number of repeaters increases, suggesting that this family of repeater networks becomes less isotonic as the number of repeaters increases.
In contrast, the metaheuristic algorithms' relative inefficiencies all increase with increasing repeater counts before plateauing to a fixed relative inefficiency after 25 repeaters.
The relative inefficiencies observed in the plateau regions appear to decrease as the number of iterations increase for each of the metaheuristic algorithms, suggesting that these algorithms are promising for graphs with extremely large repeater counts.

\section{Conclusion}
\label{sec:conclusion}
We have introduced both heuristic and exact algorithms for the important task of pathfinding, which will form an essential component of routing in quantum networks. 
The presence of non-isotonicity in metrics for pathfinding has often been neglected, even though we find it to be a generic feature in realistic quantum networks.
Taking the non-isotonicity into consideration is not just a technicality; our algorithms improve the end-to-end secret-key rate by orders of magnitude in some cases when compared with existing state of the art.
At the same time, prior algorithms that include non-isotonicity in the optimization have super-exponential complexity~\cite{caleffi2017}, while our algorithms remain efficient, in one case effectively achieving a sublinear query count scaling in the network size.
Note, however, that this does not confer sublinear time scaling for significantly larger networks, as shown in Appendix~\ref{app:tradeoff_1000}.

We evaluated our algorithms using a particular model of entanglement distribution:
swap-ASAP repeaters distributing Werner states without distillation, where the elementary link generation probability decays exponentially with the distance.
This model captures most of the essentials of near-term quantum networks, and we have explicitly demonstrated that the secret-key rate is not isotonic in this model because of the inherent rate-fidelity tradeoff.
We expect this non-isotonicity to be generic as all practical quantum networks will need to balance the rate and quality of the delivered states, independent of the underlying model, mode of distribution, or application that the entanglement is used for.
An important question is therefore whether our algorithms can also be used for pathfinding in the context of other models.
For any model where equidistant repeater placement is near optimal, BeFS-HEURISTIC should be highly effective as its underlying principle is to bound path performance by searching over homogeneous path extensions.
We expect that such behavior is typical (e.g.~equidistant repeater placement maximizes the chain's theoretical quantum capacity~\cite{pirandola2019end}) and that the algorithm can still be used when the model is extended with, e.g., cut-offs~\cite{li2021efficient}, initial losses~\cite{rozpkedek2018parameter}, and simple distillation protocols~\cite{goodenough2024near}.
Moreover, our metaheuristic methods (simulated annealing and genetic algorithms) make no specific assumptions on the model at all, and can in fact be used for pathfinding with any non-isotonic utility function.

In a similar vein, even in more intricate network settings, e.g.~multiple simultaneous requests and/or multipartite state routing, it is expected that any practical metric will be non-isotonic. Future work should attempt to generalize our findings to these more complicated and rich settings, see e.g.~\cite{bugalho2023distributing}.

\bibliographystyle{plain}
\bibliography{pathfinding.bib}

\appendix
\section{}

\subsection{Binary Search Merit Function}
\label{app:binary_search}
As mentioned in Sec.~\ref{sec:heuristic_befs}, in Heuristic-BeFS uses a merit function that is nearly, but not actually admissible.
This merit function upper-bounds the utility that can be achieved by any path that begins with the candidate prefix and ends with a suffix of length at least $L$, the distance from the end of the candidate prefix to the destination node.
To further improve the complexity of the algorithm, we assume that the utility has a unique local maximum.
In other words, it strictly increases, then strictly decreases.

Since we have a bounded list of possible solutions, and the quality of the solutions has a unique local maximum, we can use exponential search to find the maximum.

In the exponential search algorithm, we store a triple of suffix lengths that represents our search window initialized as $[1, 2, 4]$.
If the utilities calculated with these lengths is strictly increasing, then our solution's suffix length must be at least the greatest suffix length in our triple.
Thus, we can double all of our suffix lengths to "slide" the window up.
Eventually, we should end up in a situation where the solution is strictly contained in our window.
This should lead to the middle suffix length having the greatest utility.

Then, we can recursively measure utilities to shrink our window until we find the greatest utility.
We calculate the maximum possible utility using a variant of exponential search for finding maximums
instead of roots.

See Algorithm~\ref{alg:best_case} for details on the implementation.

\newcommand{\boundof}[1]{{\mathrm{bound}(#1)}}
\begin{algorithm}
\caption{Best-Case Estimation}\label{alg:best_case}
\begin{algorithmic}
  \Require{bound is a function whose maximum value is the merit}
  \Function{merit}{$p, t$}
    \State{$a, b, c := 1, 2, 4$}
    \If{$\boundof a \geq \boundof b$}
      \State{\Return{$\boundof a$}}
    \EndIf

    \While{$\boundof c > \boundof b$}
      \State{$a, b, c = b, c, 2c$}
    \EndWhile

    \While{$a + 1 \neq b$ or $b + 1 \neq c$}
      \If{$c - b > b - a$}
        \State{$d := \lfloor \frac{b + c}{2} \rfloor$}
        \If{$\boundof d > \boundof b$}
          \State{$a, b, c := b, d, c$}
        \Else
          \State{$a, b, c := a, b, d$}
        \EndIf
      \Else
        \State{$d := \lfloor \frac{a + b}{2} \rfloor$}
        \If{$\boundof d > \boundof b$}
          \State{$a, b, c := a, d, b$}
        \Else
          \State{$a, b, c := d, b, c$}
        \EndIf
      \EndIf
    \EndWhile

    \State{\Return{$\boundof b$}}
  \EndFunction
\end{algorithmic}
\end{algorithm}

\subsection{Directionality}
\label{app:directionality}
One interesting quirk of quantum pathfinding problem is that it is symmetric with respect to swapping end nodes, but the computational costs of the algorithms that we describe are not.
We call this effect \textit{directionality}, where graphs that are searched more quickly forwards have \textit{forward directionality} and graphs that are searched more quickly backwards have \textit{reverse directionality}.
Thus we have the opportunity to swap the end nodes to the configuration that is likely to be less expensive for the current algorithm.
By choosing the correct direction accurately, the average runtime can be significantly lowered.
This motivates us to construct several algorithms to detect the directionality of the algorithm.

For each of the directionality detection algorithms, we first apply a normalization procedure that consists of 2 steps:
First, we apply a Euclidean transform (rotations and translations) on the graph so that the midpoint of the 2 endpoints is at the origin and the endpoints lie along the x-axis.
This does not change the directionality, since both the algorithms and the utility function are invariant under Euclidean transforms.
Then, we scale the graph so that the source is at $(-1, 0)$ and the destination is at $(1, 0)$.
This changes the utility of each path, but we suspect that this rarely affects the directionality because directionality appears to be based on scale-invariant qualities, like how nodes are clustered together.
In addition, since this is just a heuristic, rare errors are acceptable.
After normalization, the algorithms that we create should be anti-symmetric with respect to the x-coordinate (the preferred direction should flip if we negate the x-coordinate of each node) and symmetric with respect to the y-coordinate.

\paragraph{Mean-based}
As we found earlier, it is generally more efficient to set the destination node to be the one closer to the
majority of the nodes.
We can implement this heuristic by using the mean value of the x-coordinate of the repeater nodes.
If the mean x-coordinate is positive (respectively negative), then the repeaters tend to cluster around the
destination (respectively source) node, and the algorithm should be run forwards (respectively backwards).

\paragraph{Median-based}
Another option is to use the median x-value.
Since we want to measure overall clustering, a node that is far from both endpoints may throw off the prediction.
Since the median is more resistant to outliers, it might represent the clustering behavior more accurately.

When we apply these algorithms, we see a notable increase in performance.

\begin{table}[h!]
\centering
\begin{tabular}{|c c|} 
 \hline
 Algorithm & Average Complexity \\
 \hline
 None              & 100.520 \\ 
 Multi-Thread Race & 163.232 \\
 Mean              & 97.418 \\
 Median            & 96.582 \\
 \hline
\end{tabular}
\vspace{2mm}
\caption{Comparison of directionality estimation algorithms.}
\label{table:1}
\end{table}

As a potential future extension, the prediction algorithms could include additional features like graph distances and spectral analysis. In addition, these values could be fed into a machine learning model.

\subsection{Bounds on the Performance of Repeater Chains} \label{app:bounds}

In this appendix we prove that the performance of a swap-ASAP repeater chain as described in Sec.~\ref{sec:problem_statement} can be bounded by assuming that some of the repeater nodes are spaced equidistantly and have perfect quantum memory.
This result forms the basis of our BeFS-EXACT algorithm and is stated at the end of this appendix in Corollary~\ref{cor:skr_of_extension_best_for_smallest_length}.
Proving this requires showing that the rate of entanglement generation is optimized by placing repeaters equidistantly, even when the repeater locations can only be freely chosen along a part of the chain and the remaining part of the chain does not have its repeaters distributed homogeneously.
This result is stated in Corollary~\ref{cor:rate_optimal_for_hom_extension}, while most of the statistical heavy lifting required to show this is taken care of by Theorem~\ref{thm:hom_dominates}.
Extending the result to not just the entangling rate but also the secret-key rate, by assuming that the repeaters that are freely placed have perfect quantum memory, is achieved mostly by Lemma~\ref{lem:skr_optimal_for_hom_extension}.
This appendix contains moreover one final result, Lemma~\ref{lem:max_N}, that helps us limit the search space for finding the optimal number of repeaters in the BeFS algorithms.
For the scaffolding used to prove the results in this appendix, we have been inspired by the work done in~\cite{coopmans2022}.

First we introduce some notation, and in particular the concept of statistical dominance.
We then lead up to Theorem~\ref{thm:hom_dominates} which shows that the maximum over a number of geometric distributions is stochastically dominant when those geometric distributions are homogeneous.

\begin{definition} \label{def:cdf}
Let $X$ be a random variable.
The cumulative distribution function (CDF), denoted $F_X(x)$, is defined by
\begin{equation}
F_X(x) = \text{Pr}(X \leq x).
\end{equation}
\end{definition}

\begin{definition} \label{def:stochastic_dominance}
Let $X$ and $Y$ be two random variables.
We say that $X$ stochastically dominates $Y$ and write $X \geq_\text{st} Y$ if
\begin{equation} \label{eq:stochastic_dominance}
\text{Pr}(X > z) \geq \text{Pr}(Y > z)
\end{equation}
for all $z$.
\end{definition}

\begin{lemma} \label{lem:expected_values_from_dominance}
Let $X$ and $Y$ be two discrete random variables, both with domain $\mathbb N$ (the natural numbers excluding 0).
If $X \geq_\text{st} Y$, then $\mathbb E[X] \geq \mathbb E[Y]$.
\end{lemma}
\begin{proof}
Because $\text{Pr}(X > z) = \sum_{z' = z+1}^\infty \text{Pr}(X = z')$ we have
\begin{equation}
\sum_{z=1}^\infty \text{Pr}(X > z) = \sum_{z=1}^\infty z \text{Pr}(X=z) = \mathbb E[X].
\end{equation}
Therefore, if $\text{Pr}(X > z) \geq \text{Pr}(Y>z)$ for all $z$, then it follows immediately that $\mathbb E[X] \geq \mathbb E[Y]$.
\end{proof}

\begin{lemma} \label{lem:dominance_in_terms_of_cdfs}
Let $X$ and $Y$ be two random variables.
Then $X \geq_\text{st} Y$ if and only if
\begin{equation}
F_X(z) \leq F_Y(z)
\end{equation}
for all $z$.
\end{lemma}
\begin{proof}
This follows directly from Eq.~\eqref{eq:stochastic_dominance} because $\text{Pr}(X > z) = 1 - F_X(z)$ (and similar for $Y$).
\end{proof}

\begin{lemma} \label{lem:cdf_of_max}
Let $X = \max{\left(\left\{X_i\right\}_{i=1}^n\right)}$ for some independent random variables $X_i$.
Then its CDF can be expressed as
\begin{equation} \label{eq:cdf_of_max}
F_X(x) = \prod_i F_{X_i}(x).
\end{equation}
\end{lemma}
\begin{proof}
The value of $X$ is smaller than or equal to $x$ if and only if the same is true for each $X_i$, hence,
\begin{equation}
\begin{aligned}
\text{Pr}(X \leq x) &= \text{Pr}(X_1 \leq x \wedge X_2 \leq x \wedge ... \wedge X_n \leq x)\\
&= \text{Pr}(X_1 \leq x) \text{Pr}(X_2 \leq x) ... \text{Pr}(X_n \leq x),
\end{aligned}
\end{equation}
where the second line follows because the $X_i$ are all independent.
Eq.~\eqref{eq:cdf_of_max} then follows directly from Def.~\ref{def:cdf}.
\end{proof}

\begin{theorem} \label{thm:hom_dominates}
Let $\mathcal X$ be the family of random variables
\begin{equation} \label{eq:family_for_hom_dominates_thm}
\mathcal X = \left\{ \max \left( \left\{X_i\right\}_{i=1}^n \right) | X_i \sim \text{geom}(p_i), \prod_i^n {p_i} = p_\text{tot}  \right\}, \nonumber
\end{equation}
where $0 < p_\text{tot}, p_i \leq 1$.
Let $X_\text{hom} \in \mathcal X$ be the member of that family for which $p_1 = p_2 = ... = p_n$, i.e.,
$X_\text{hom} = \max\left( \left\{ X_{\text{hom}, i}  \right\}_{i=1}^n \right)$ with $X_{\text{hom}, i} \sim \text{geom}(\sqrt[n]{p_\text{tot}})$.
Then, for every $X \in \mathcal X$ it holds that
\begin{equation}
X \geq_\text{st} X_\text{hom}.
\end{equation}
\end{theorem}
\begin{proof}
We can parameterize all members of the family $\mathcal X$ in terms of parameters $\{\Delta_i\}_{i=1}^{n-1}$ such that $X = X(\Delta_1, \Delta_2, ..., \Delta_{n-1}) = \max \left( \left\{X_i\right\}_{i=1}^n \right)$ with $X_i \sim \text{geom}(p_i)$, $p_i = \sqrt[n]{p_\text{tot}} e^{\Delta_i}$ for $i = 1, 2, ..., n-1$ and $p_n = \sqrt[n]{p_\text{tot}} e^{- \sum_i \Delta_i}$, where $\Delta_i \leq \frac 1 n \ln{(1/p_\text{tot})}$ such that $0<p_i\leq1$ and $\sum_i \Delta_i \geq - \frac 1 n \ln{(1/p_\text{tot})}$ such that $0<p_n\leq1$.

By Lemma~\ref{lem:cdf_of_max} we can write the CDF of $X$, which is implicitly also a function of the $\Delta_i$, as
\begin{equation}
F_X(k) = \prod_{i=1}^n F_{X_i}(k),
\end{equation}
where for the geometric random variables $X_i$ we have
\begin{equation}
F_{X_i}(k) = 1 - (1 - p_i)^k,
\end{equation}
and where $k = 1, 2, ...$ as the domain of a geometric random variable is the positive integers.
We will show that, for each $k$, $F_X(k)$ attains its maximum value when $\Delta_i=0$ for all $i$, i.e., when all $p_i$ are equal.
This proves the theorem per Lemma~\ref{lem:dominance_in_terms_of_cdfs}.

First, we note that when $p_\text{tot} = 1$, there is only one allowed parameter configuration, namely $\Delta_i=0$ for all $i$;
the theorem is hence trivially true in that case.
Throughout the rest of the proof, we will assume $p_\text{tot} < 1$.
Moreover, we also note that $F_X(k)$ is trivially maximal for $\Delta_i =0$ when $k=1$.
In that case we have $F_{X_i}(k) = p_i$ and hence $F_X(k) = \prod_i p_i = p_\text{tot}$.
Hence, the function is constant and attains its maximum value everywhere in its domain, including when all $p_i$ are equal.

In order to prove that $F_X(k)$ is maximal for $k = 2, 3, ...$ when all $p_i$ are equal, we show that for any $i$, when keeping all $\Delta_j$ for $j \neq i$ fixed, choosing $\Delta_i$ such that $p_i = p_n$ strictly maximizes $F_X(k)$.
It then follows that if for any $i < n$ it holds that $p_i \neq p_n$, the value of $F_X(k)$ can be strictly improved by changing the value of $\Delta_i$ such that $p_i$ and $p_n$ become equal.
Therefore, whenever any $p_i$ is not equal to $p_n$, $F_X(k)$ is not maximal, and it thus follows that if $F_X(k)$ is maximal then $p_i = p_n$ for all $i < n$, and hence all $p_i$ are equal.
The method we will use to show that a strict maximum along the one-dimensional slice of $F_X(k)$ associated with $\Delta_i$ is located at $p_i = p_n$ is by showing that there is exactly one critical point (i.e., location where the derivative w.r.t. $\Delta_i$ is zero) along that slice, and that this critical point is a strict local maximum.
Because the slice is one dimensional, the critical point must then necessarily also be a strict global maximum (where ``global'' refers only to that particular slice).
More specifically, such one-dimensional slices are defined by $- \frac 1 n \ln{(1/p_\text{tot})} + \sum_{j \neq i} \Delta_j \leq \Delta_i \leq \frac 1 n \ln{(1/p_\text{tot})}$ and we will show that the strict maximum occurs at $\Delta_i^* = -\frac 1 2 \sum_{j \neq i} \Delta_j$.

To that end, we now compute the derivatives of $F_X(k)$.
Writing $\partial_i = \frac {\partial}{\partial \Delta_i}$ and suppressing the argument $k$ in the CDFs, we have
\begin{equation}
\begin{aligned}
\partial_i F_X = \left(F_{X_n}\partial_i F_{X_i}  + F_{X_i} \partial_i F_{X_n} \right) \prod_{j \neq i, n} F_{X_j}.
\end{aligned}
\end{equation}
Moreover, for $j \in \{i, n\}$, we have
\begin{equation}
\partial_i F_j = k(1-p_j)^{k-1}\partial_i p_j = \pm k(1-p_j)^{k-1} p_j,
\end{equation}
where the $+$ sign holds for $j=i$ and the $-$ sign holds for $j=n$.
Defining
\begin{equation}
\begin{aligned}
T:=& \left( 1 - (1 - p_n)^k \right)(1 - p_i)^{k-1} p_i -&  \\
 &\left(1 - (1 - p_i)^k \right)(1 - p_n)^{k-1} p_n \\
\end{aligned}
\end{equation}
we then find
\begin{equation}
\partial_i F_X = k T \prod_{j \neq i, n} F_{X_j}.
\end{equation}
Since $k$ and all the $F_{X_j}$ are strictly positive, it now suffices to study $T$ in order to find and analyze critical points.
To this end we introduce $r = 1 - p_n$ and $s = 1 - p_i$, satisfying $0 \leq r, s < 1$, such that
\begin{equation}
\begin{aligned}
T &= (1 - r^k) s^{k-1} (1 - s) - (1 - s^k) r^{k-1} (1 - r)\\
&= (r^k - s^k) - (r^{k-1} - s^{k-1}) - (rs)^{k-1} (r - s).
\end{aligned}
\end{equation}
We immediately note that if $s=0$, $T=0$ if and only if $r=0$, and vice versa.
The condition $r=s=0$ (or $p_i=p_n=1$) is only true if $\Delta_i = \frac 1 n \ln{(1/p_\text{tot})}$ and $\sum_{j \neq i} \Delta_j = - \frac 2 n \ln{(1/p_\text{tot})}$.
But the one-dimensional slice for which this last condition holds only consists of a single point (so it is, in fact, a zero-dimensional slice), and hence $F_X$ is always trivially maximal along that slice.
Having covered the cases $s=0$ and $r=0$, we will continue below for $0<s,r<1$.

Now, we use the identity
\begin{align}
W_k &:= \sum_{i=0}^{k-1} r^{k-1-i}s^i \\
r^k - s^k &= (r - s) W_k
\end{align}
to rewrite
\begin{equation}
T = (r - s) \left( W_k - W_{k-1} - (rs)^{k-1}\right).
\end{equation}
Since we have $W_k = sW_{k-1} + r^{k-1}$, this is equivalent to
\begin{equation}
T = (r - s) \left( W_{k-1} (s - 1) - r^{k-1} (s^{k-1} - 1) \right).
\end{equation}
We can now use the same identity (or, equivalently, the formula for the value of a geometric series) to find
\begin{equation}
s^{k-1} - 1 = (s - 1) \sum_{i=0}^{k-2} s^i.
\end{equation}
Hence,
\begin{equation}
\begin{aligned}
T &= (r - s) (s - 1) \left(W_{k-1} - r^{k-1} \sum_{i=0}^{k-2} s^i \right) \\
&= (r - s) (s - 1) \sum_{i=0}^{k-2} s^i (r^{k-2-i} - r^{k-1}) \\
&= (s - r) (1 - s) r^{k-1} \sum_{i=0}^{k-2} s^i (r^{-(1+i)} - 1).
\end{aligned}
\end{equation}
Now, remembering that $0 < r, s < 1$, we have $1 - s > 0$, $s^i > 0$, and $r^{-(1+i)} - 1 > 0$.
Therefore, $T = 0$ if and only if $s - r = 0$, i.e., if $p_i = p_n$, and hence the only critical point along the one-dimensional slice of $F_X$ associated with $\Delta_i$ occurs when that condition is true, which is at $\Delta_i^* = - \frac 1 2 \sum_{j \neq i} \Delta_j$.
Moreover, we note that the sign of $T$, and hence of $\partial_i F_X$, is equal to the sign of $s - r = p_n - p_i$, and that moreover, $p_n - p_i$ is a function that strictly decreases when $\Delta_i$ increases (it is equal to $-2 \sqrt[n]{p_\text{tot}} e^{\Delta_i^*} \sinh(\Delta_i - \Delta_i^*)$).
Hence the derivative is strictly negative for $\Delta_i > \Delta_i^*$, while the derivative is strictly positive for $\Delta_i < \Delta_i^*$.
This indicates that $\Delta_i^*$ is a strict local maximum, and being the only critical point, the strict global maximum of the slice.
It then follows from the arguments presented above that $F_X$ is strictly maximal if and only if all $p_i$ are equal.
\end{proof}

Geometric distributions can be used to represent the time required to create entanglement over an elementary link of a repeater chain.
Therefore we can use Theorem~\ref{thm:hom_dominates} to make statements about which repeater chains can distribute entanglement the fastest.
It directly follows that homogeneous repeater chains have the highest rate, as shown in Corollary~\ref{cor:rate_optimal_for_hom_chain}.

\begin{corollary} \label{cor:rate_optimal_for_hom_chain}
Let $R$ be the rate of a repeater chain as modeled in Sec.~\ref{sec:problem_statement}, with $N + 2$ total nodes in the chain.
Let the end nodes (nodes 0 and $N + 1$) be distance $L$ apart, and let the $N$ repeater nodes be placed along a line between the end nodes.
Denote the distance between nodes $i$ and $i + 1$ as $\ell_i > 0$ for $i = 0, ..., N$, which satisfy $\sum_i \ell_i = L$.
Then $R$ is maximal when the repeaters are placed equidistantly, i.e., when $\ell_i = \frac L {N+1}$.
\end{corollary}
\begin{proof}
Let $T$ be the time required to create end-to-end entanglement.
Then, by definition, the entangling rate is $R = 1/\mathbb E[T]$.
Hence, $R$ is maximal for equidistantly spaced repeaters if and only if $\mathbb E[T]$ is minimal in that case.
In the model described in Sec.~\ref{sec:problem_statement} we have
\begin{equation}
T = t_\text{att} X,
\end{equation}
where $t_\text{att}$ is the attempt duration (which is the same for all edges as the attempts are synchronized to match the longest link), $X = \max{\left(\left\{X_i\right\}\right)}$ is the number of rounds until completion, $X_i \sim \text{geom}(p_i)$ is the number of rounds required to create entanglement on edge $i$ and $p_i = 10^{-\frac{\alpha}{10} \ell_i}$ is the success probability of entanglement generation over edge $i$.

The attempt duration in this model is given by
\begin{equation}
t_\text{att} = \frac 1 c \max{\left(\left\{\ell_i\right\}_i\right)}
\end{equation}
with $c$ the speed of light in fiber.
This is trivially minimized by splitting the total length $L$ into equal parts.

Now, we observe that we have
\begin{equation} \label{eq:ptot_in_terms_of_L}
\prod_i p_i = \ 10^{- \frac \alpha {10} \sum_i \ell_i} = 10^{- \frac \alpha {10} L} =: p_\text{tot}.
\end{equation}
Hence, the number of attempts $X$ fulfills the assumptions of Theorem~\ref{thm:hom_dominates}, and it follows that the instantiation of $X$ for which all $p_i$ are equal, i.e., for which the $\ell_i$ equally partition the total length $L$, is stochastically dominated by all other instantiations of $X$.
Lemma~\ref{lem:expected_values_from_dominance} then implies that $\mathbb E[X]$ is minimal when the nodes are placed equidistantly.

As equidistant node placement minimizes both $t_\text{att}$ and $\mathbb E[X]$ it also minimizes $\mathbb E[T] = t_\text{att} \mathbb E[X]$.
\end{proof}

That the rate of a homogeneous chain is optimal has been known, see, e.g.,~\cite{avis2024}.
However, saying that one random variable statistically dominating another is a stronger statement than saying that its expected value is larger.
This allows us to show that it is optimal to place repeaters equidistantly even if only some of the repeater locations can be freely chosen.
We now build towards this result, which is contained in Corollary~\ref{cor:rate_optimal_for_hom_extension}.

\begin{lemma} \label{lem:domination_of_max}
Let $X$, $X'$ and $Y$ be random variables such that $X \geq_\text{st} X'$.
Then
\begin{equation}
\max{ \left( \left\{ X, Y \right\} \right)} \geq_\text{st} \max{ \left( \left\{ X', Y \right\} \right)}
\end{equation}
\end{lemma}
\begin{proof}
Let $Z = \max{ \left( \left\{ X, Y \right\} \right)}$ and $Z' = \max{ \left( \left\{ X', Y \right\} \right)}$.
By Lemma~\ref{lem:cdf_of_max} we have
\begin{equation}
F_Z(z) = F_X(z) F_Y(z)
\end{equation}
and
\begin{equation}
F_{Z'}(z) = F_{X'}(z) F_Y(z).
\end{equation}
Because $X \geq_\text{st} X'$, by Lemma~\ref{lem:dominance_in_terms_of_cdfs} we have $F_{X}(z) \leq F_{X'}(z)$ for all $z$.
Since $F_Y(z) \geq 0$ (it is a probability), it follows that for all $z$ it holds that $F_{Z}(z) \leq F_{Z'}(z)$ and hence (again by Lemma~\ref{lem:dominance_in_terms_of_cdfs}) $Z \geq_\text{st} Z'$.
\end{proof}

\begin{corollary} \label{cor:hom_dominates_in_max}
Let $\mathcal Z$ be the family of random variables
\begin{equation}
\mathcal Z = \left\{ \max \left( \left\{Y\right\} \cup \left\{X_i\right\}_{i=1}^n \right) | X_i \sim \text{geom}(p_i), \prod_i^n {p_i} = p_\text{tot}  \right\},\nonumber
\end{equation}
where $Y$ is a discrete random variable and $0 < p_\text{tot}, p_i \leq 1$.
Let $Z_\text{hom} \in \mathcal X$ be the member of that family for which $p_1 = p_2 = ... = p_n$, i.e.,
$Z_\text{hom} = \max\left( \left\{Y \right)\} \cup \left\{ X_{\text{hom}, i}  \right\}_{i=1}^n \right)$ with $X_{\text{hom}, i} \sim \text{geom}(\sqrt[n]{p_\text{tot}})$.
Then, for every $Z \in \mathcal Z$ it holds that
\begin{equation}
Z \geq_\text{st} Z_\text{hom}.
\end{equation}
\end{corollary}
\begin{proof}
Because $\max \left(\{A, B, C\} \right) \nolinebreak[4]=\nolinebreak[4]\max (\{A, \max(\{B, C\})\})$ we can rewrite
\begin{equation}
\mathcal Z = \left\{ \max \left( \left\{Y, X \right\}\right)| X \in \mathcal X\right\},
\end{equation}
where $\mathcal X$ is as defined in Eq.~\eqref{eq:family_for_hom_dominates_thm}.
We know from Theorem~\ref{thm:hom_dominates} that, for all $X \in \mathcal X$, $X \geq_\text{st} X_\text{hom}$, where $X_\text{hom} = \max\left( \left\{ X_{\text{hom}, i}  \right\}_{i=1}^n \right)$.
Then it follows by Lemma~\ref{lem:domination_of_max} that, for all $X \in \mathcal X$, $\max \left(\left\{Y, X \right\} \right) \geq_\text{st} \max \left(\left\{Y, X_\text{hom} \right\} \right) = \mathcal Z_\text{hom}$.
\end{proof}

\begin{corollary} \label{cor:rate_optimal_for_hom_extension}
Let $R$ be the rate of a repeater chain as modeled in Sec.~\ref{sec:problem_statement} of $M + N + 1$ nodes.
Let the nodes 1, 2, ..., $M$ be at fixed locations, with fixed distances between them.
Let the nodes $M$ and $M + N + 1$ be distance $L$ apart, and let the $N$ repeater nodes between those nodes be placed along a straight line.
Denote the distance between nodes $M + i - 1$ and $M + i$ as $\ell_i \geq 0$ for $i = 1, 2, ..., N+1$, which satisfy $\sum_i \ell_i = L$.
Then $R$ is maximal when the repeaters are placed equidistantly, i.e., when $\ell_i = \frac L {N+1}$.
\end{corollary}
\begin{proof}
Just as in the proof of Corollary~\ref{cor:rate_optimal_for_hom_chain}, we can write $R = 1/\mathbb E[T]$ where $T = t_\text{att} X$ where $t_\text{att}$ is the attempt duration and $X$ is the number of rounds required until completion.
We will show that $\mathbb E[T]$ is minimal when repeaters are placed equidistantly by showing that both $t_\text{att}$ and $\mathbb E[X]$ are minimal, which implies that $R$ is maximal.
To this end, let us assign the index $i$ to the edge between nodes $i$ and $i + 1$ for $i = 1, 2, ..., M+N$, and let $l_i$ be the length of edge $i$.
Note that $l_{i+ M - 1} = \ell_i$ for $i = 1, 2, ..., N+1$, while $l_i$ is a constant for $i = 1, 2, ..., M-1$.

We start with the attempt duration.
Then
\begin{equation} \label{eq:attempt_duration_as_split_max}
\begin{aligned}
t_\text{att} &= \max\left(\left\{l_i\right\}_{i=1}^{M+N}\right) \\
&= \max \left( \left\{ \max\left(\left\{l_i\right\}_{i=1}^{M-1}\right),  \max\left(\left\{\ell_i \right\}_{i=0}^{N}\right) \right\} \right).
\end{aligned}
\end{equation}
Now, because $\max\left(\left\{\ell_i \right\}_{i=0}^{M+1}\right)$ is minimal when the total length $L$ is split equally among the $\ell_i$ it follows that $t_\text{att}$ is minimal when the repeaters are placed equidistantly.

For the number of rounds until completion, we have
\begin{equation} \label{eq:number_of_rounds_split_up}
X = \max \left(\left\{X_i\right\}_{i=1}^{M+N}\right) = \max \left( \left\{A, B\right\} \right)
\end{equation}
where $X_i$ is the number of rounds until edge $i$ has successfully created entanglement, $A = \max \left( \left\{X_i\right\}_{i = 1}^{M-1}\right)$, and $B = \max \left( \left\{X_i\right\}_{i = M}^{M+N}\right)$.
We have $X_i \sim \text{geom}(p_i)$ with $p_i = e^{-\frac \alpha {10} l_i}$.
Hence, $\prod_{i=M}^{M+N} p_i = 10^{-\frac \alpha {10} L} =: p_\text{tot}$, and therefore $X$ fulfills the assumptions of Corollary~\ref{cor:hom_dominates_in_max}.
The instantiation of $X$ for which the repeaters are placed equidistantly is thus stochastically dominated by all other instantiations and hence, by Lemma~\ref{lem:expected_values_from_dominance}, $\mathbb E[X]$ is minimal in that case.
\end{proof}

Above we have shown that, when extending a repeater chain, the entangling rate $R$ is optimized when that extension contains equidistantly placed repeaters.
However, the utility function considered in this paper is the secret-key rate SKR, and therefore what we need are results about how to extend a repeater chain such that the secret-key rate is optimal.
This is difficult due to the complicated nature of noise incurred during the storage of qubits
This problem can be solved by assuming the extension added to the repeater chain has perfect quantum memory (but still imperfect state generation).
We show that homogeneous extensions are optimal in that case in Lemma~\ref{lem:skr_optimal_for_hom_extension}, and then extend this in Corollary~\ref{cor:skr_of_extension_best_for_smallest_length} to also show that the optimal extension is as short as possible.
This is useful because, for BeFS-EXACT, we are interested in a bound on the best achievable SKR, and assuming that some quantum memories are perfect can only increase the SKR.

\begin{lemma} \label{lem:skr_optimal_for_hom_extension}
Let SKR be the secret-key rate of a repeater chain as modeled in Sec.~\ref{sec:problem_statement} of $M + N + 1$ nodes.
Let the nodes 1, 2, ..., $M$ be at fixed locations, with fixed distances between them.
Let the nodes $M$ and $M + N + 1$) be distance $L$ apart, and let the $N$ repeater nodes between those nodes be placed along a straight line.
Denote the distance between nodes $M + i- 1$ and $M + i$ as $\ell_i \geq 0$ for $i = 1, 2, ..., N + 1$, which satisfy $\sum_i \ell_i = L$.
Let moreover $F$ be the elementary-link fidelity, and $T$ the coherence time of nodes $1, 2, ..., M-1$, while nodes $M, M+1, ..., M + N + 1$ have perfect quantum memory.
Then SKR is maximal when the repeaters are placed equidistantly, i.e., when $\ell_i = \frac L {N+1}$.
\end{lemma}
\begin{proof}
As defined in Eq.~\eqref{eq:SKR}, SKR is the product of the entangling rate $R$ and the secret-key fraction SKF.
$R$ is maximal when repeaters are placed equidistantly by Corollary~\ref{cor:rate_optimal_for_hom_extension}.
Below, we will prove that SKF as well is maximal when all $\ell_i$ are equal.
It then follows that SKR is maximal when the repeaters are placed equidistantly.
For this proof we will use the same edge labeling as in the proof of Corollary~\ref{cor:rate_optimal_for_hom_extension}: edge $i$ is the edge between nodes $i$ and $i + 1$ and has length $l_i$.

First, let us define the depolarizing channel
\begin{equation}
\mathcal D_{k, q}(\rho) = q\rho + (1 - q) \frac{\mathbb 1_k}{2} \text{Tr}_k(\rho),
\end{equation}
where $q$ is the depolarizing parameter corresponding to the channel, $k$ is the index of the qubit the channel acts on, $\mathbb 1_k$ is the identity operator on the Hilbert space of qubit $k$, and $\text{Tr}_k$ denotes the partial trace with respect to the Hilbert space of qubit $k$.
A convenient property of this channel is that concatenating multiple depolarizing channel gives a new depolarizing channel with the product of the depolarizing parameters:
\begin{equation}
\mathcal D_{k, q_2} \otimes \mathcal D_{k, q_1} = \mathcal D_{k, q_1q_2}.
\end{equation}
Moreover, let us define the Werner state with Werner parameter $w$ as
\begin{equation} \label{eq:Werner_state}
W_w = w \ket{\phi^+}\bra{\phi^+} + (1 - w) \frac{\mathbb 1}{4}.
\end{equation}

An important property of the depolarizing channel for this proof is that depolarization on the first qubit of a Bell state has the same effect as depolarization of the second Bell state; in both cases a Werner state will be created with a Werner parameter that is equal to the depolarizing parameter,
\begin{equation}
\mathcal D_{1, q}(\ket{\phi^+}\bra{\phi^+}) = \mathcal D_{2, q}(\ket{\phi^+}\bra{\phi^+}) = W_q.
\end{equation}
This allows us to apply the ``ping-pong trick''~\cite{wilde2013quantum}, i.e., move noise around between qubits.
A direct consequence is that when performing entanglement swapping between states $W_{w_1}$ and $W_{w2}$, we get a new Werner state $W_{w1 w_2}$, which can be seen as follows:
the Werner states can be considered Bell states with depolarizing noise on the qubits that are not involved in the entanglement-swapping operation.
Because the swap and the noise now act on different parts of the Hilbert space, they must commute.
Hence, the result can be calculated by first calculating the result of the swap, which creates a perfect Bell state (up to a Pauli correction that commutes with all noise and can be accounted for at the end nodes), and then applying the noise channels.
The noise channels can moreover be moved to the same qubit, allowing us to use the multiplicative property of the depolarizing channel.

In the model under consideration the states that are created upon a successful entanglement-generation attempt are $W_{w_0}$ with
\begin{equation}
w_0 = \frac{4F - 1}{3}.
\end{equation}
Hence, if there were no memory noise, the end-to-end state created by performing entanglement swapping between all copies of $W_{w_0}$ created on the $M + N$ edges would be $W_{w_0^{M + N}}$.
The effect of the memory noise is that for the first repeaters, i.e., the nodes $i = 2, 3, ..., M-1$ there is a depolarizing channel with depolarizing parameter $\exp(-(|t_{i}-t_{i-1}|/T)$ (see Eq.~\eqref{eq:memory_noise}), where $t_i$ is the completion time of link $i$ (remember that the other nodes are assumed to have perfect memory, and that end nodes do not incur decoherence as they directly measure their qubits).
Because we can simply ping-pong these depolarizing channels around, the final result will be the state $W_{w_{e2e}}$ with
\begin{equation}
w_{e2e} = w_0^{M + N} \prod_{i=2}^{M-1} e^{\frac{-|t_{i}-t_{i-1}|}{T}}.
\end{equation}
This is a random variable, as the $t_i$ are random variables; the Werner parameter of the state corresponding to when no postselection on the completion times takes place is $\mathbb E[w_{e2e}]$.

The important thing now is to note that only the completion times of the first $M$ links appear, and hence the state does not depend on the completion time of links $M, M+1, ..., M+N$.
That seems to suggest that the state is fully independent of the parameters $\ell_i$, but that is not necessarily true as each $t_i$ is proportional to the attempt duration, which in this model is determined by the longest link in the chain.
If one of the $\ell_i$ exceeds the lengths of all the first $M$ edges, those first edges will be slowed down accordingly, and hence the final state that is produced is noisier.
We can make this explicit by writing $t_i = t_\text{att} X_i$ where $t_\text{att}$ is the attempt duration and $X_i$ is geometrically distributed.
Then we have
\begin{equation} \label{eq:we2e}
w_{e2e} = w_0^{M + N} \prod_{i=2}^{M-1} \left(e^{-|X_{i}-X_{i-1}|}\right)^{\frac{t_\text{att}}{T}}.
\end{equation}
Because $\exp(-|X_{i+1} - X_i|) \leq 1$ the effect of increasing $t_\text{att}$ is to apply a non-increasing map to each point in the domain of $w_{e2e}$, and hence $\mathbb E[w_{e2e}]$ is a non-increasing function of $t_\text{att}$.

Now, if all the $N$ repeaters are placed equidistantly, the maximum length among the edges of the second leg of the repeater chain is minimal, and hence the maximum length among all the edges in the chain is minimal to the extent that it can be influenced by tuning the $\ell_i$ parameters.
This, in turn, means that equidistant placement leads to minimal $t_\text{att}$ (see also the proof of Corollary~\ref{cor:rate_optimal_for_hom_extension}) and hence results in a state with minimal noise.
To show formally that this maximizes the SKF, note that for a Werner state $W_w$ the QBERs are $Q_X = Q_Z = (1 - w) / 2$ and hence
\begin{equation}
\text{SKF} = \max\left(0, 1 - 2h\left( \frac{1-\mathbb E[w_{e2e}]}{2} \right) \right),
\end{equation}
where $h$ is the binary entropy function.
This is a non-decreasing function of $\mathbb E[w_{e2e}]$ on its domain $0 \leq \mathbb E[w_{e2e}] \leq 1$, thereby completing the proof that the SKF is maximal for equidistant repeater placement.
\end{proof}

Now it only remains to show that, given equidistant repeater placement in the extension, the SKR is optimized by making the total length $L$ of the extension as small as possible.

\begin{lemma} \label{lem:dominant_when_ptot_is_smallest}
Let $X(p_\text{tot}, n) = \max\left( \left\{ X_i  \right\}_{i=1}^n \right)$ where $X_i \sim \text{geom}(\sqrt[n]{p_\text{tot}})$, for $n \in \mathbb N$ and $0 < p_\text{tot} \leq 1$.
If $p_\text{tot}' \geq p_\text{tot}$, then
\begin{equation}
X(p_\text{tot}, n) \geq_\text{st} X(p_\text{tot}', n).
\end{equation}
\end{lemma}
\begin{proof}
The CDF of each of the $X_i$ is given by $F_{X_i}(k) = 1 - (1 - \sqrt[n]{p_\text{tot}})^k$.
By Lemma~\ref{lem:cdf_of_max} we then have
\begin{equation}
F_{X}(k) = \left(1 - (1 - \sqrt[n]{p_\text{tot}})^k\right)^n.
\end{equation}
This is a function that monotonically increases with $p_\text{tot}$ for all $0 < p_\text{tot} \leq 1$ (this follows from $x^y$ being an increasing function of $x$ for any $x, y>0$ and $1 - x$ being a decreasing function of $x$).
Hence, $X(p'_\text{tot}, n)$ has a larger CDF than $X(p_\text{tot}, n)$ throughout.
Stochastic domination then follows from Lemma~\ref{lem:dominance_in_terms_of_cdfs}.
\end{proof}

\begin{corollary} \label{cor:skr_of_extension_best_for_smallest_length}
Let SKR be the secret-key rate of a repeater chain as modeled in Sec.~\ref{sec:problem_statement} of $M + N + 1$ nodes.
Let the nodes 1, 2, ..., $M$ be at fixed locations, with fixed distances between them.
Let the $N$ repeater nodes between nodes $M$ and $M + N + 1$ be placed along a curve that can be freely chosen but that can be no shorter than $L_\text{min}$.
Denote the distance between nodes $M + i- 1$ and $M + i$ as $\ell_i \geq 0$ for $i = 1, 2, ..., N + 1$, which satisfy $\sum_i \ell_i \geq L_\text{min}$.
Let moreover $F$ be the elementary-link fidelity, and $T$ the coherence time of nodes $1, 2, ..., M-1$, while nodes $M, M+1, ..., M + N + 1$ have perfect quantum memory.
Then SKR is maximal when the length of the curve is minimized, i.e., $\sum_i \ell_i = L_\text{min}$, and the repeaters are placed equidistantly, i.e., $\ell_i = \frac {L_\text{min}} {N+1}$.
\end{corollary}
\begin{proof}
Let $\sum_i \ell_i = L$.
Then, for any given value of $L$, by Lemma~\ref{lem:skr_optimal_for_hom_extension}, the SKR is maximal for $\ell_i = \frac L {N + 1}$.
It then only remains to prove that the optimal value for $L$ is $L_\text{min}$.
As in the proof of Lemma~\ref{lem:skr_optimal_for_hom_extension} we note that SKR is the product of SKF and $R$.
We will first show that making $L$ as small as possible maximizes SKF, and then that it maximizes $R$.
It then follows that SKR is maximal for $L = L_\text{min}$.

From Eq.\eqref{eq:we2e} we see that $L$ can only affect SKF through its effect on the attempt duration, $t_\text{att}$.
We can use Eq.~\eqref{eq:attempt_duration_as_split_max} to write
\begin{equation}
t_\text{att} = \max \left( \left\{ \max\left(\left\{l_i\right\}_{i=1}^{M-1}\right),  \frac L {N + 1} \right\}\right),
\end{equation}
where the $l_i$ for $i= 1, 2, ..., M-1$ are the constant lengths of the first $M - 1$ edges.
From this relation, it follows that $t_\text{att}$ is minimal when $L$ is minimal.
Moreover, we have concluded in the proof of Lemma~\ref{lem:skr_optimal_for_hom_extension} that the SKF is maximal when $t_\text{att}$ is minimal, and hence $L = L_\text{min}$ maximizes SKF.

Now, to prove that the rate $R$ is maximal we write $R = 1/\mathbb[T]$ with $T = t_\text{att} X$ where $X$ is the number of rounds until completion, which is the maximum of the number of rounds until completion over all of the $M + N$ edges.
We saw in Eq.~\eqref{eq:number_of_rounds_split_up} we have $X = \max \left( \left\{A, B\right\} \right)$ where $A$ is independent of $L$ and, since the repeater nodes are equidistantly spaced, $B = \max\left( \left\{ X_i  \right\}_{i=1}^n \right)$ where $X_i \sim \text{geom}(\sqrt[n]{p_\text{tot}})$ and $p_\text{tot} = 10^{-\frac \alpha {10} L}$ (see Eq.~\eqref{eq:ptot_in_terms_of_L}).
From Lemma~\ref{lem:dominant_when_ptot_is_smallest} we know that the $B$ for which $p_\text{tot}$ is smallest is stochastically dominated by all other instantiations of $B$, and by Lemma~\ref{lem:domination_of_max} the same is true for $X$.
As $p_\text{tot}$ decreases monotonically with $L$, it follows that $X$ is stochastically dominated for $L = L_\text{min}$ by any other instantiation of $X$.
By Lemma~\ref{lem:expected_values_from_dominance} it then follows that $\mathbb E[X]$ is minimal for $L = L_\text{min}$.
Since we already concluded that $t_\text{att}$ is minimal in that case, it follows that $R$ is maximal.
\end{proof}

Finally, we include a result that bounds the number of repeaters that can be included in a chain such that there is a nonzero secret-key rate.
This helps finding the optimal number of repeaters in the BeFS-EXACT and BeFS-HEURISTIC algorithms more quickly.

\begin{lemma} \label{lem:max_N}
Let SKR be the secret-key rate of a repeater chain as modeled in Sec.~\ref{sec:problem_statement} with $N + 1$ nodes and elementary-link fidelity $F$.
Then, for every $N \geq N^*$ it holds that $\text{SKR} = 0$, where
\begin{equation}
N^* = \frac{ \log \left(1 - 2Q^* \right) } {\log \left( \frac {4F - 1}{3} \right) }
\end{equation}
and $Q^* \approx 0.110028$ is the unique solution on the domain $0 \leq Q \leq \frac 1 2$ of the equation $h(Q) = \frac 1 2$ where $h(Q) = -Q\log2(Q) - (1 - Q)\log2(1 - Q)$ is the binary entropy function.
\end{lemma}
\begin{proof}
Let $Q_X = Q_Z = Q$ be the QBER, where the X and Z QBERs are equal because end-to-end entangled states are always Werner states and these are invariant when rotating both qubits simultaneously between the X and Z bases (i.e., under the unitary $H \otimes H$, where $H$ is the Hadamard gate).
The QBER is defined as the probability of obtaining anti-correlated instead of correlated measurement results.
Sticking to the Z basis for simplicity, for a Werner state $W_w$ (defined in Eq.~\eqref{eq:Werner_state} this is given by
\begin{equation}
Q = \bra{01} W_w \ket{01} + \bra{10} W_w \ket{10} = \frac {1 - w} 2.
\end{equation}
Because $0 \leq w \leq 1$, we have $0 \leq Q \leq \frac 1 2$.

Now, the SKF is given by (see Eq.~\eqref{eq:SKF}) $\text{SKF} = \max(0, 1 - 2h(Q))$.
The binary entropy function $h(Q)$ is strictly increasing on $0 \leq Q \leq \frac 1 2$ with $h(0) = 0$ and $h(\frac 1 2) = 1$, and hence on that domain there is a unique value $Q^*$ such that $h(Q^*) = \frac 1 2$.
It follows that for all $Q \geq Q^*$ the SKF is zero.
Numerically, it can be determined that $Q^* \approx 0.110028$.

We now assume that there is no memory noise in the repeater chain, hence providing a lower bound on the QBER.
In that case, the end-to-end state is $W_{w_{e2e}}$ with $w_{e2e} = w_0^N$ where $w_0 = (4F-1)/3$ (see Eq.~\eqref{eq:we2e} and surrounding discussion).
The SKF will be zero whenever $w_{e2e} \leq 1 - 2Q^*$, and hence when $N \geq \log_{w_0}(1 - 2Q^*) = \log(1 - 2Q^*) / \log((4F - 1)/3)$.
\end{proof}

\subsection{Validity of the Complexity Estimate}
\label{app:complexity validity}
\begin{figure}[t!]
\centering
\includegraphics[width=.99\columnwidth]{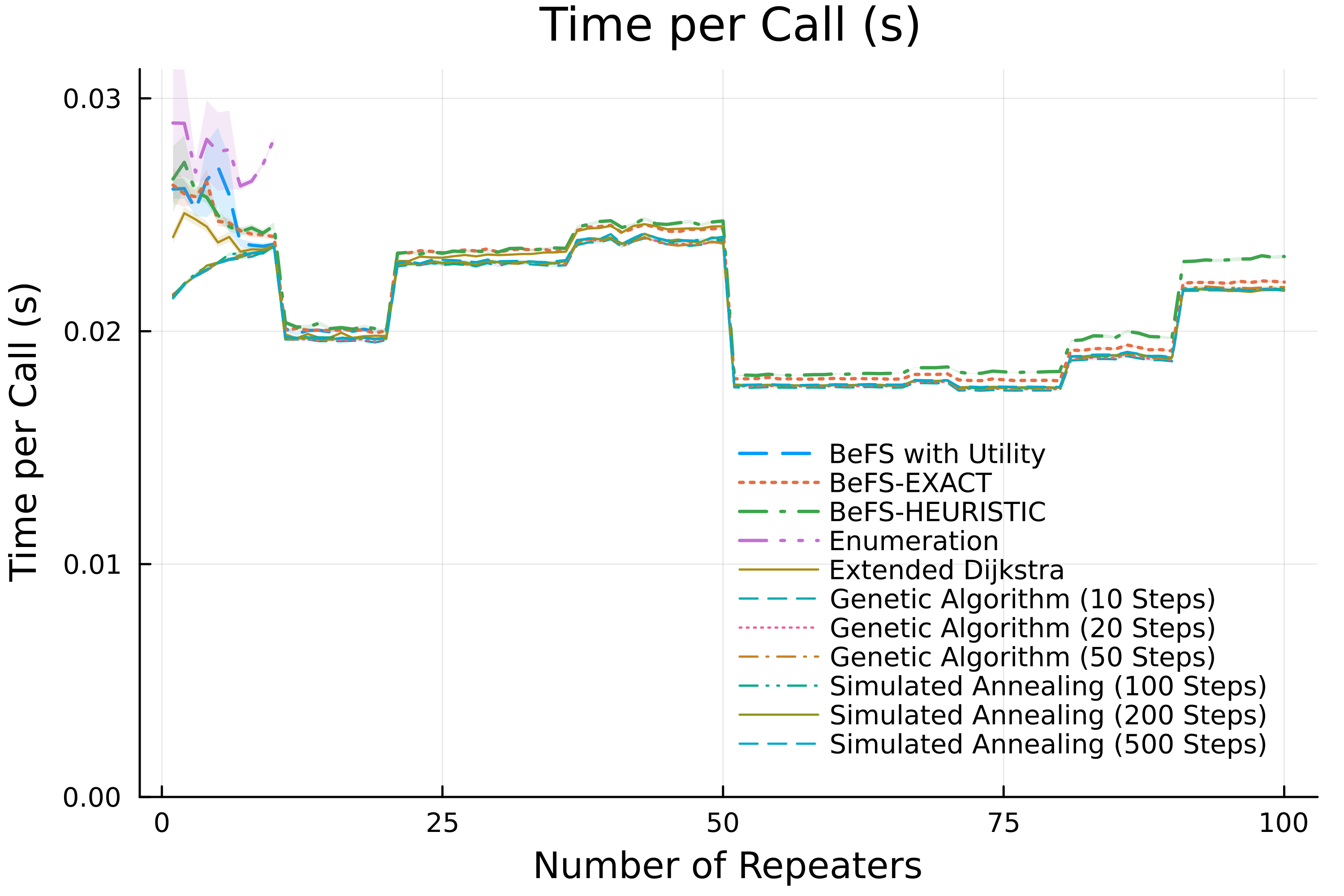}\caption{
    A comparison of the different algorithms for time per secret key rate estimate.
    The x-axis has the number of repeaters, the y-axis has the complexity.
    \label{fig:call_efficiency}
}
\end{figure}
As shown in Fig.~\ref{fig:call_efficiency}, all of the algorithms presented have runtimes that are nearly exactly proportional to the number of SKR estimates.
There is a major spike for very low costs, but this is primarily due to the warm-up time of Julia's JIT compiler.
There is also a cliff at 10 repeaters, which reflects the fact that this was run on a different computer.

\subsection{Removal of Unnecessary Edges}
\label{app:bcc_pruning}
In principle, we can safely remove every edge in the network that is not part of the best path without changing the solution.
However, this is difficult to calculate.
A more conservative approach, however, can be implemented efficiently.
Since the solution path is a path, any edge that is not on any path between the end nodes ($s$ and $t$) can be removed without changing the solution to the pathfinding problem.
We can efficiently identify these edges using the properties of biconnected components.
A biconnected component (BCC) of an undirected graph is an equivalence class of the edges of a graph, where every pair of edges is in the same BCC iff there is a cycle connecting them.
If we add a new edge between $s$ and $t$, then the cycles containing this edge are the paths from $s$ to $t$ with this edge appended.
Thus, the BCC containing this edge is exactly the set of edges that lie on some path from $s$ to $t$.
Biconnected components can be calculated in $O(|V| + |E|)$ time using Tarjan's algorithm~\cite{tarjan1972}.
Thus, before running each of our algorithms, we use Tarjan's algorithm to identify and remove unused edges.

\subsection{Monte-Carlo Error of Exact and Near-Exact Algorithms}
\label{app:monte_carlo_error}
\begin{figure}[t!]
\centering
\includegraphics[width=.99\columnwidth]{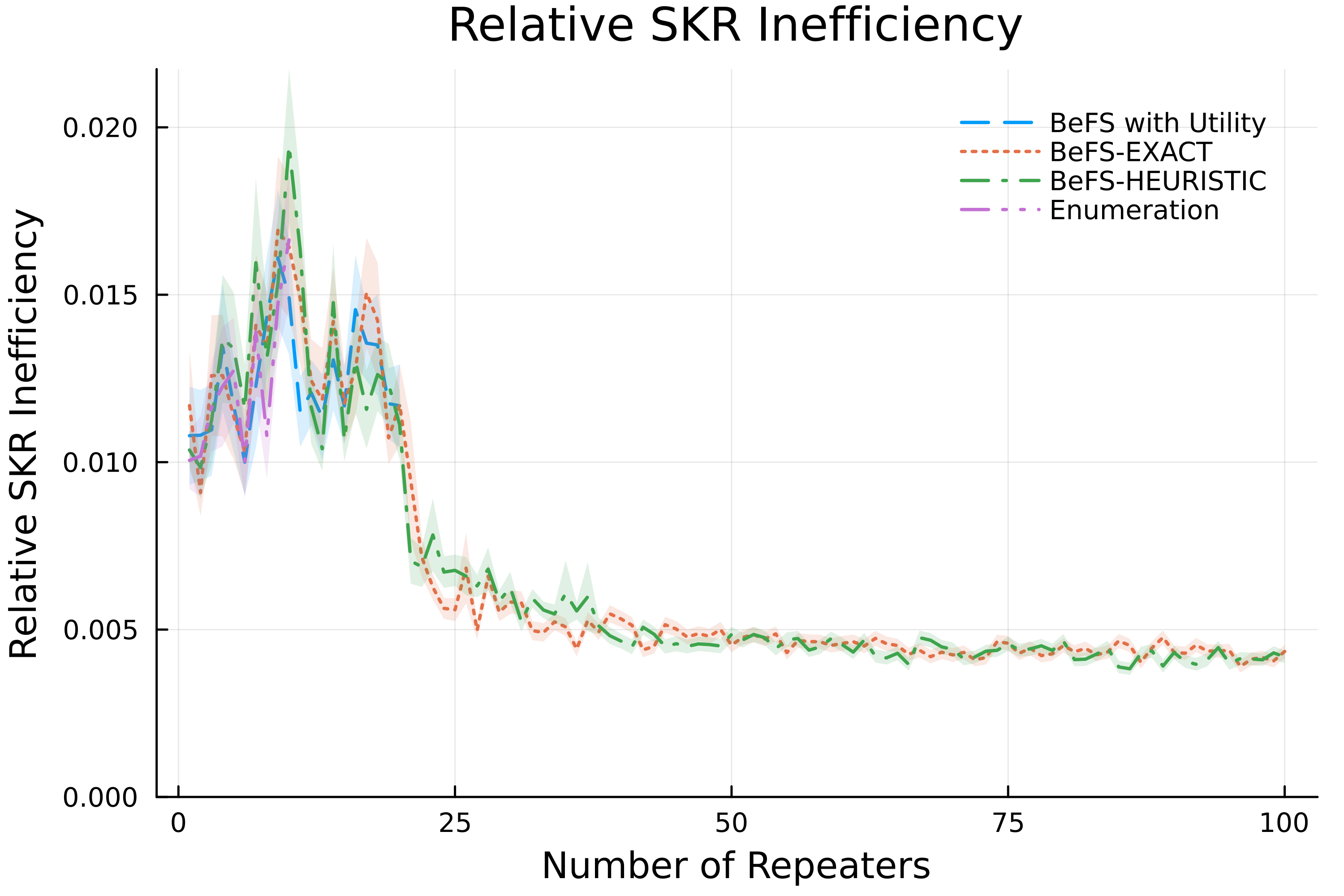}
\caption{
    A comparison of the exact algorithms for utility imperfection as a fraction of the best algorithm for that graph.
    The x-axis has the number of repeaters, the y-axis defect fraction.
    \label{fig:num_rep_vs_skr_deficit_exact}
}
\end{figure}
As shown in Fig.~\ref{fig:num_rep_vs_skr_deficit_exact}, all of the exact and near-exact algorithms have roughly the same deficit ratio.
The relative SKR inefficiency is calculated in the way described in Sec.~\ref{sec:algorithm_comparisons}, with the relative SKR inefficiency being the ratio of the SKR deficit against the best algorithm.

We suspect that these algorithms almost always produce the best path, and that the measured inefficiency is a result of noise in the SKR estimation algorithm.
This is because the measured inefficiency is roughly the same for all of the algorithms, which reflects the symmetry in how inefficiency is calculated.
If there was error stemming from any of the algorithms themselves, as opposed to error in the SKR estimation function, we would expect that algorithm to have regions with a higher inefficiency than the other algorithms.

The relative SKR inefficiency is low for small repeater counts, since many graphs lack a viable path due to memory coherence issues, which forces all of the relative inefficiencies to be 0.
The relative SKR inefficiency is also low for high repeater counts, since it is divided by the maximum secret-key rate, which increases with repeater count.
This leaves the middle, with the relative SKR inefficiency peaking at 1.5\% around 10--20 repeaters.

\subsection{SKR vs. Query Complexity Tradeoff for Large Networks}
\label{app:tradeoff_50}
\begin{figure}[t!]
\centering
\includegraphics[width=.99\columnwidth]{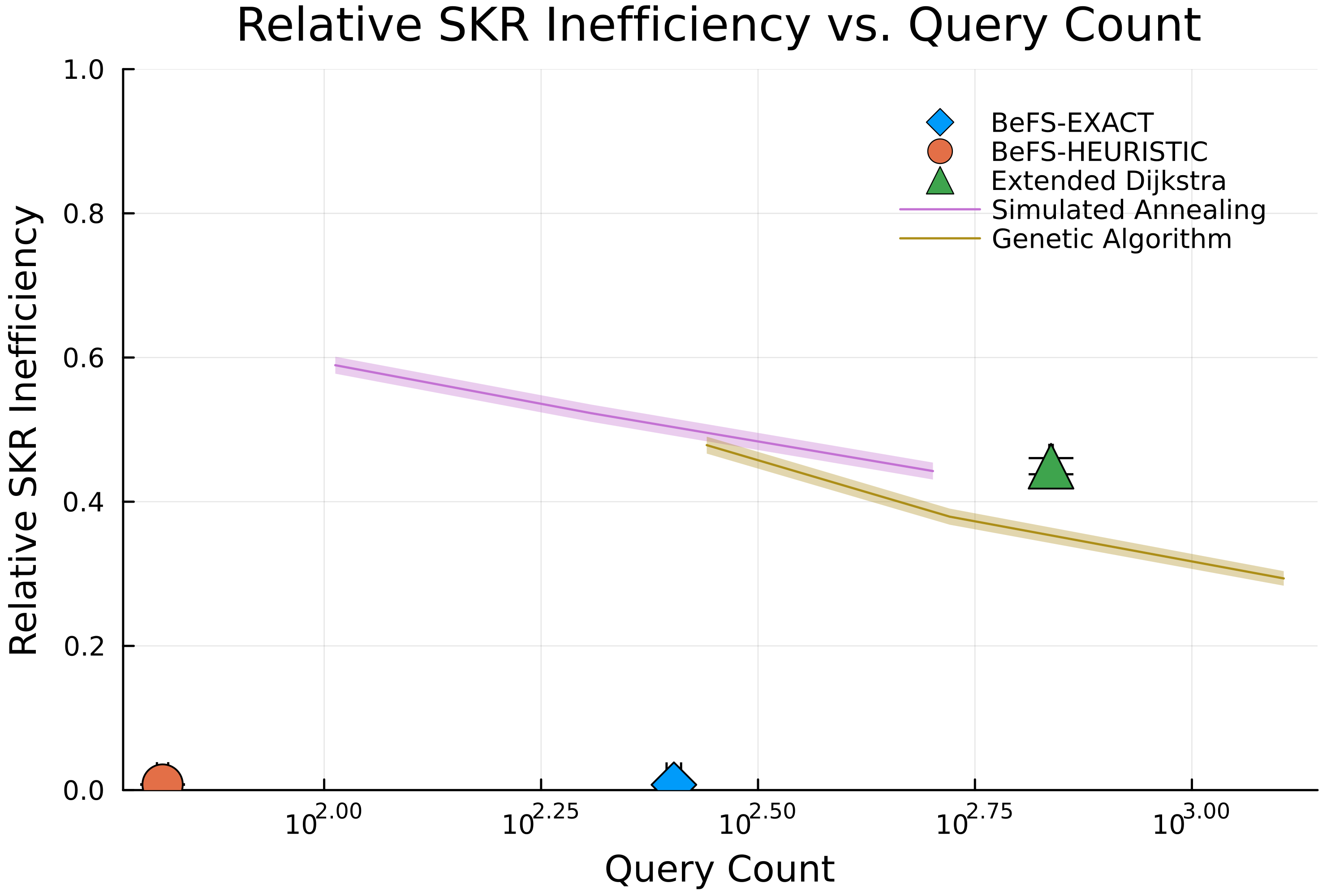}
\caption{
    A scatterplot of the query count of the algorithm (x-axis) vs. the relative secret key rate inefficiency (y-axis) when benchmarking the algorithms described in this paper on a set of randomly generated graphs.
    The query count is the number of times that the SKR estimation algorithm is called.
    The algorithms were run on a set of 1000 randomly generated Waxman graphs that each have 50 repeaters.
    The graphs were generated with the parameters $\alpha=0.5, \beta=0.9, L=300\mathrm{km}$.
    The metaheuristic algorithms (simulated annealing and genetic algorithms) are parameterized by the number of steps, so they are plotted as lines rather than points.
    The simulated annealing algorithms are plotted for 100 to 500 steps, and the genetic algorithms are plotted for 10 to 50 generations.
    Note that this is plot was generated with a subset of the data used for Fig.~\ref{fig:num_rep_vs_query_complexity} and Fig.~\ref{fig:num_rep_vs_skr_deficit} rather than a special run, meaning that fewer algorithms are plotted here.
    \label{fig:query_complexity_vs_skr_deficit_50}
}
\end{figure}
In comparison to Fig.~\ref{fig:query_complexity_vs_skr_deficit}, Fig.~\ref{fig:query_complexity_vs_skr_deficit_50} shows the extended Dijkstra's algorithm to be strictly worse than some parameterizations of simulated annealing and genetic algorithms.
In addition, the BeFS algorithms remain accurate.

\subsection{Time Tradeoff for Large Networks}
As an extension of our experiments, we also benchmarked the most promising algorithms from our investigations on networks of up to 1000 repeaters.
We show our results in Fig.~\ref{fig:num_rep_vs_time_1000} and Fig.~\ref{fig:num_rep_vs_query_complexity_1000}.
\label{app:tradeoff_1000}
\begin{figure}[ht!!!]
\centering
\includegraphics[width=.99\columnwidth]{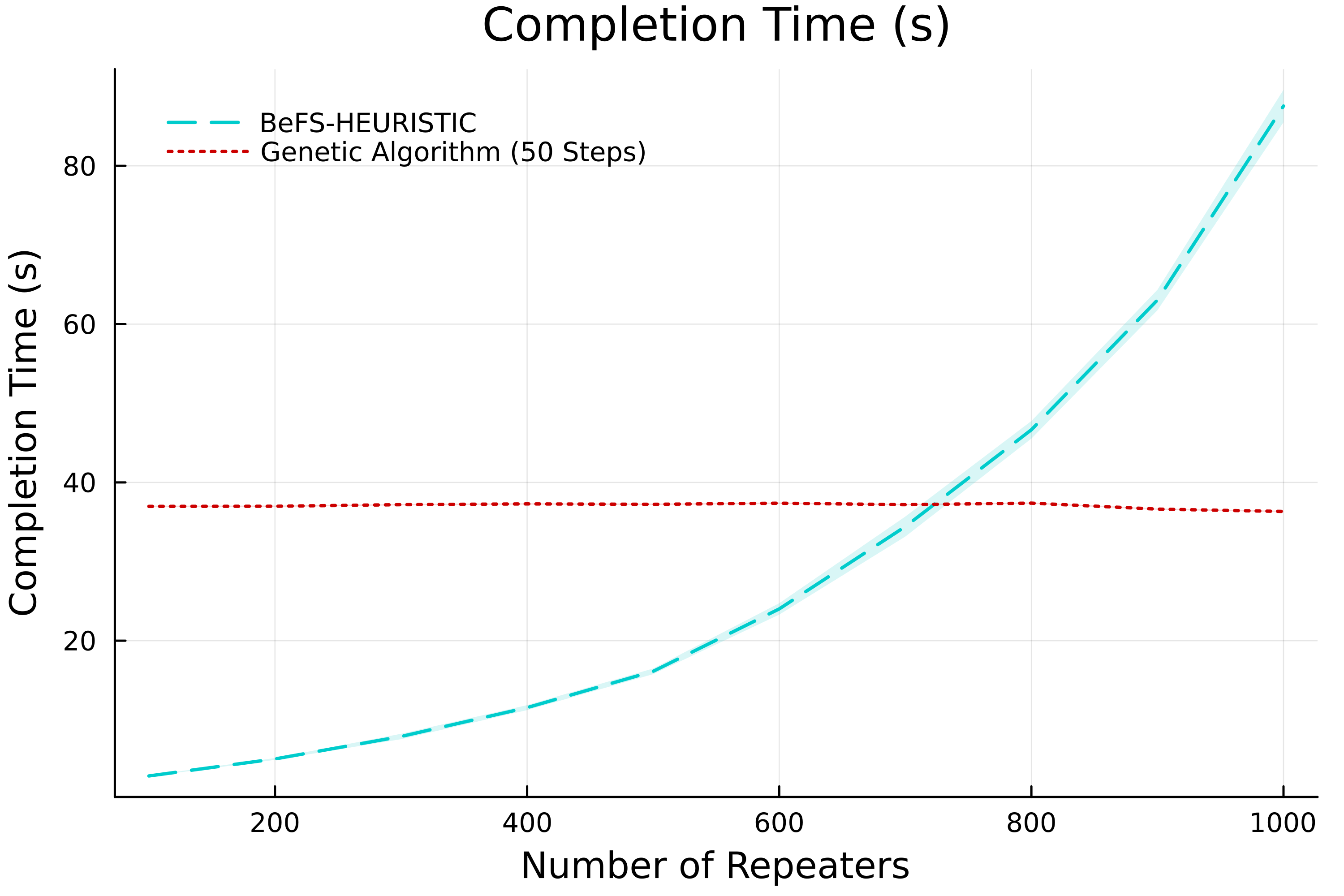}
\caption{
    A plot comparing the average runtime for BeFS-HEURISTIC v.s. Genetic Algorithms with 50 generations for graphs with 100--1000 repeaters.
    \label{fig:num_rep_vs_time_1000}
}
\end{figure}

\begin{figure}[ht!!!]
\centering
\includegraphics[width=.99\columnwidth]{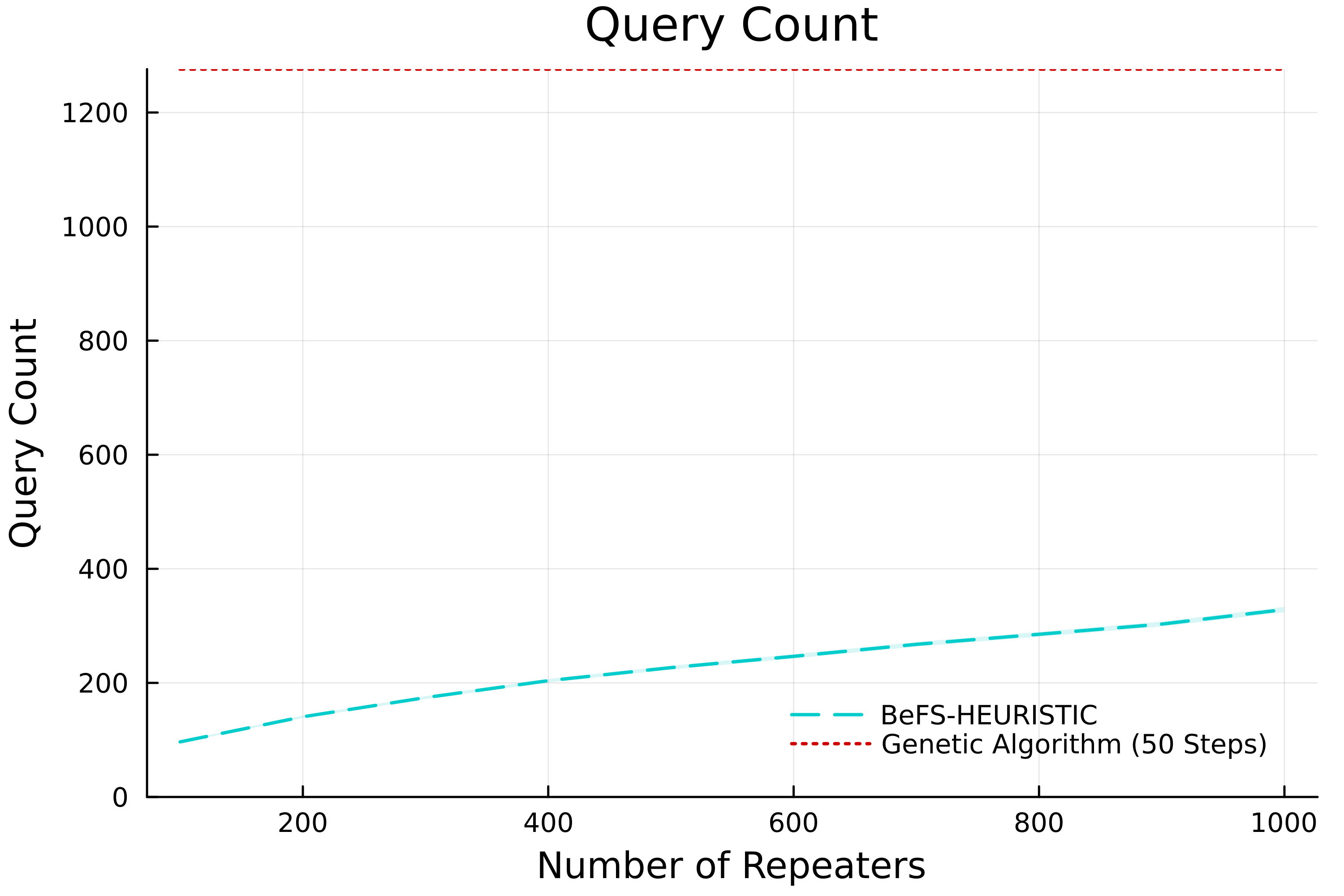}
\caption{
    A plot comparing the average query count for BeFS-HEURISTIC v.s. Genetic Algorithms with 50 generations for graphs with 100--1000 repeaters.
    \label{fig:num_rep_vs_query_complexity_1000}
}
\end{figure}

    

\newpage
Comparing Fig.~\ref{fig:num_rep_vs_time_1000} and Fig.~\ref{fig:num_rep_vs_query_complexity_1000}, we see that the query complexity does not reflect the actual runtime behavior of BeFS-HEURISTIC for large repeater networks.
This confirms our hypothesis that the cost of dominance checks becomes the dominant runtime cost for large repeater networks.

\end{document}